\documentclass[final,12pt]{colt2018} 

\usepackage{mathrsfs,bm,bbm,tikz}

\def \bP {\mathbb{P}}
\def \bE {\mathbb{E}}

\def \cM {\mathcal{M}}
\def \cL {\mathcal{L}}
\def \spo {\mathsf{Poi}}
\newcommand{\reals}{\mathbb{R}}

\newcommand{\calD}{{\mathcal{D}}}

\newcommand{\calM}{{\mathcal{M}}}

\newcommand{\calP}{{\mathcal{P}}}

\newcommand{\calV}{{\mathcal{V}}}

\newtheorem{question}{Question}
\newcommand{\stepa}[1]{\overset{\rm (a)}{#1}}
\newcommand{\stepb}[1]{\overset{\rm (b)}{#1}}
\newcommand{\stepc}[1]{\overset{\rm (c)}{#1}}

\title[Local Moment Matching]{Local moment matching: A unified methodology for symmetric functional estimation and distribution estimation under Wasserstein distance}
\usepackage{times}

\coltauthor{\Name{Yanjun Han} \Email{yjhan@stanford.edu}\\
	\addr Department of Electrical Engineering, Stanford University
\AND
	\Name{Jiantao Jiao} \Email{jiantao@stanford.edu}\\
 \addr Department of Electrical Engineering, Stanford University
 \AND
 \Name{Tsachy Weissman} \Email{tsachy@stanford.edu}\\
 \addr Department of Electrical Engineering, Stanford University
 }

\begin{document}

%


%
\maketitle

\begin{abstract}
We present \emph{Local Moment Matching (LMM)}, a unified methodology for symmetric functional estimation and distribution estimation under Wasserstein distance. We construct an efficiently computable estimator that achieves the minimax rates in estimating the distribution up to permutation, and show that the plug-in approach of our unlabeled distribution estimator is ``universal" in estimating symmetric functionals of discrete distributions. Instead of doing best polynomial approximation explicitly as in existing literature of functional estimation, the plug-in approach conducts polynomial approximation implicitly and attains the optimal sample complexity for the entropy, power sum and support size functionals.
\end{abstract}


\section{Introduction and Main Results}
Given $n$ independent samples from a discrete distribution $P=(p_1,\cdots,p_S)$, we aim to estimate the distribution vector $P$ up to permutation. In other words, let $P^<=(p_{(1)}, p_{(2)}, \cdots, p_{(S)})$ be the sorted version of $P$ (i.e., $p_{(1)}\le p_{(2)}\le \cdots\le p_{(S)}$ are the order statistics of $P$), we would like to find an estimator $\hat{P}$ which comes close to minimizing the sorted $\ell_1$ distance
\begin{align*}
\bE_P \lVert \hat{P}-P^<\rVert_1
\end{align*}
in the minimax sense.

Our study of estimating the sorted distribution $P^{<}$ is motivated by the following facts:
\begin{enumerate}
	\item The sorted distribution $P^{<}$ can be interpreted as the distribution $P$ up to permutation, or the multiset of probabilities in $P = (p_1,p_2,\ldots,p_S)$, or the ``tail'' of a distribution. In economics, the theory of \emph{long tail}~\cite{anderson2004long} emphasizes the significance of products in the \emph{tail}, and inferring the sorted distribution precisely shows the shape of the tail.
	\item Estimating the sorted distribution $P^{<}$ turns out to require significantly less number of samples than that required to estimate the distribution $P$ under the same $\ell_1$ loss, as shown by~\cite{Valiant--Valiant2011} using a different Wasserstein loss function. 
	\item The sorted distribution estimate proves to be useful in estimating symmetric functionals of the distribution, which are defined as functionals of $P$ that can also be viewed as functionals of $P^{<}$. Indeed, \cite{Valiant--Valiant2011,Valiant--Valiant2013estimating, acharya2016unified} constructed estimators of the sorted distribution can be used to plug-in some symmetric functionals to achieve the information theoretic limits in certain parameter regimes, which performed significantly better than the approach of plugging-in the empirical distribution.
	\item The whole distribution $P$ can be decomposed into two parts: the sorted distribution $P^{<}$ and the permutation from $P^{<}$ to $P$. Being able to design computationally efficient schemes to achieve the information theoretic limit in estimating $P^{<}$ sheds light on the general question of inferring parameters up to group transformations, which is of fundamental significance in statistics and machine learning~\cite{kong2017spectrum,tian2017learning}. 
\end{enumerate}

The main idea to solve the traditional distribution estimation problem (i.e., estimating $P$) is to use the empirical frequency and/or its variants, which turn out to be minimax optimal for various loss functions including $\ell_2$ \cite{Steinhaus1957problem,Trybula1958problem,Rutkowska1977minimax,Olkin1979admissible}, $\ell_1$ \cite{daskalakis2012learning,Diakonikolas2014beyond,han2015minimax,kamath2015learning} and KL loss \cite{kamath2015learning}. To consistently estimate $P$, usually it is required to observe each symbol $i\in [S]$ sufficiently many times on average; for example, $n\gg S$ is a necessary and sufficient condition for the existence of an estimator which estimates $P$ within a vanishing $\ell_1$ error \cite{han2015minimax}. However, recent studies suggested that estimating $P^<$ might be significantly easier than estimating $P$: compared with an oracle with the same observation $X_1,\cdots,X_n$ and perfect knowledge of $P^<$, there still exists some estimator which performs nearly as well as the oracle even if $S=\infty$ under the $\ell_1$ loss \cite{valiant2015instance} and KL loss \cite{orlitsky2015competitive}. This observation shows that the ``labeling" from $P^<$ to $P$ is the difficult step in estimating $P$, and estimating the sorted distribution $P^<$ may only require sub-linear samples (i.e., $n\ll S$).

Two main approaches have been proposed in literature to estimate the sorted distribution $P^{<}$. One is the approach of \emph{profile maximum likelihood (PML)}~\cite{orlitsky2004modeling,acharya2009recent}, which aims at solving the sorted distribution that maximizes the likelihood of observing the \emph{sorted} empirical distribution. It is not clear how to solve the corresponding optimization problem efficiently. Algorithms that approximately solve the PML have been proposed in the literature, including~\cite{orlitsky2004modeling,vontobel2012bethe,pavlichin2017approximate}, without clear theoretical approximation guarantees. It was shown in~\cite{acharya2016unified} that plugging-in the profile maximum likelihood distribution into a variety of symmetric functionals (namely, entropy, support size, support coverage, and distance to uniformity) achieves the information theoretic limit when the number of samples is not ``too'' large. 

A different approach, which initiated from~\cite{Efron--Thisted1976}, proposed to use linear programming to find a sorted distribution that was consistent with the observed frequency
counts. This approach was adapted and rigorously analyzed in~\cite{Valiant--Valiant2011,Valiant--Valiant2013estimating} under a Wasserstein distance loss function, where it was shown that plugging-in the inferred sorted distribution from the linear program into certain symmetric functionals (namely, entropy, support size, support coverage, and distance to uniformity) results in estimators that achieve the information theoretic limit in the constant error regime. 

Various questions remain unsolved given existing literature. It is not clear how to efficiently provably solve the PML, and the proof of the optimality of PML in the plug-in machinery of symmetric functional estimation heavily relies on the fact that the observations can only take values in a finite set. It does not apply to the Gaussian setting, where one observes a Gaussian random vector $X\sim \mathcal{N}(\mu, I_p)$, and would like to estimate the sorted version of the mean vector $\mu$. For the linear programming approach, it was not shown to achieve the optimal dependence on $\epsilon$ in entropy estimation, and it was not clear whether it is near-optimal if we plug it in other functionals. Indeed, the general achievability proof is done through a Lipschitz continuity argument in~\cite{Valiant--Valiant2011}, and it was not clear whether they can match the lower bounds for individual functionals. 

The main mathematical reason that motivated this paper is to develop relations between estimation of (nonsmooth) functionals of distributions, and estimation of the sorted distribution. In the first realm, ~\cite{Lepski--Nemirovski--Spokoiny1999estimation} considered the problem of $L_r$ norm estimation in Gaussian noise model and utilized Fourier approximation theory, while~\cite{Cai--Low2011} considered estimating the $\ell_1$ norm of normal mean and applied best polynomial approximation. The work of~\cite{Valiant--Valiant2011power} developed ``Chebyshev bump'' based approximation and proposed linear estimators that achieve the optimal dependence on $\epsilon$ in estimating the entropy, distance to uniformity, and support size when the sample size $n$ is not too ``large''. Minimax rates for estimation of a large variety of functionals were solved in the past few years, including entropy~\cite{wu2016minimax,Jiao--Venkat--Han--Weissman2015minimax}, R\'enyi entropy~\cite{Acharya--Orlitsky--Suresh--Tyagi2014complexity}, support size~\cite{wu2015chebyshev}, support coverage~\cite{orlitsky2016optimal}, distinct elements~\cite{wu2016sample}, $L_1$ distance~\cite{jiao2016minimax}, Kullback--Leibler divergence~\cite{bu2016estimation,Han--Jiao--Weissman2016minimaxdivergence}, squared Hellinger divergence~\cite{Han--Jiao--Weissman2016minimaxdivergence}, $\chi^2$ divergence~\cite{Han--Jiao--Weissman2016minimaxdivergence}, support coverage from multiple populations~\cite{raghunathan2017estimating}, $L_r$ norm of a regression function in Gaussian white noise~\cite{Han--Jiao--Mukherjee--Weissman2017adaptive}, and differential entropy~\cite{han-jiao-weissman-wu2017minimax}. The latest batch of work have developed essentially a framework of proving minimax upper and lower bounds for functional estimation problems, which was called the \emph{Approximation} approach in~\cite{Jiao--Venkat--Han--Weissman2015minimax}. The main idea is, we first use concentration inequalities to ``zoom in'' sets that are guaranteed to contain the true parameters with overwhelming probability, and then apply unbiased estimators of (best) approximation polynomials up to a certain degree in those sets. The minimax lower bounds are proved using the dual representation of best polynomial approximation over each individual sets that we may ``zoom in''. For a crisp illustration of the lower bound technique, we refer the readers to~\cite{jiao2017minimax}. 

The \emph{Approximation} approach requires to compute deterministic approximations for each individual functional separately, and is naturally a \emph{non-plug-in} approach. It was shown in~\cite{jiao2017minimax} that indeed any plug-in approach cannot hope to completely replace the \emph{Approximation} approach: there exist certain functionals such that \emph{any} plug-in approach fails to achieves the statistical limit.

This paper aims at bridging the \emph{Approximation} approach illustrated above and the plug-in approach. We start with the following question. 

\begin{question}\label{question}
	Find an estimator of the sorted distribution that satisfies the following properties:
	\begin{enumerate}
		\item Plugging-in the estimator into a large variety of symmetric functionals achieves the information theoretic limit;
		\item It has a clear mathematical correspondence with the \emph{Approximation} approach (hence generalizable to Gaussian settings);
		\item It achieves the minimax rates in estimating sorted distribution under $\ell_1$ loss;
		\item It is efficiently computable.
	\end{enumerate}
\end{question}

We present \emph{Local Moment Matching \rm (LMM)}, an approach that provably answers the question above. 

The intuition behind LMM, on the highest level, is the following. It shares with \emph{Approximation} the same first step, which is to use concentration inequalities to ``zoom in'' the smallest sets that are guaranteed to contain the true parameters with overwhelming probability. The difference appears in the second step, while \emph{Approximation} tries to find a polynomial that closely approximates the functional over the specific set and then use unbiased estimators to estimate the polynomial, LMM aims to find a set of numbers, whose total number equals to the number of true parameters we believe are in the set, whose \emph{moments} match the \emph{unbiased} estimates of the moments of true parameters. In other words, \emph{Approximation} conducts an \emph{explicit} approximation of functional, and LMM conducts and \emph{implicit} approximation, since LMM needs to achieve statistical optimality for a large variety of symmetric functionals.  

The idea of combining moment matching and linear programming appeared before in~\cite{kong2017spectrum,tian2017learning}. However, the key contribution of our work is the feature of \emph{local} rather than \emph{global} moment matching in~\cite{kong2017spectrum,tian2017learning}. The advantage and necessity of \emph{locality} can be seen from the following thought experiment. Suppose $p_1 = 0.1, p_2 = 0.1 + \epsilon, p_3 = 0.3$. Given sufficiently many number of samples, it is easy to infer that the probability corresponding to symbol $3$ is larger than the probability corresponding to symbol $1$, but it may be unclear whether $p_1 \geq  p_2$ or $p_1 \leq p_2$ for $\epsilon$ small enough. The \emph{global} moment matching approach tries to find a sorted distribution that matches the moments of $(p_1,p_2,p_3)$, while \emph{local} moment matching tries to only match the moments corresponding to $(p_1,p_2)$. In other words, the \emph{local} approach only uses linear program and moment matching when there is ambiguity about the relative magnitude of the true probabilities, while \emph{global} moment matching discards the information we already know (such as $p_3 > p_1$) thus behaves sub-optimally unless the number of samples is very small. Indeed, it is the reason why~\cite{tian2017learning} is only statistically optimal for very small $n$ ($n\leq \ln S$ in the notation of this paper). We also mention that the setting in~\cite{tian2017learning} is not entirely identical to ours, since in~\cite{tian2017learning} $\sum_{i = 1}^S p_i$ is not necessarily one. 

We present our main results below. 
\begin{theorem}\label{thm.main}
	For $n\gtrsim \frac{S}{\ln S}$, we have\footnote{Notation $a_n=\tilde{\Theta}(b_n)$ means that for any $\epsilon>0$, we have $n^{-\epsilon}b_n\ll a_n\ll n^{\epsilon}b_n$.}
	\begin{align*}
	\inf_{\hat{P}}\sup_{P\in\cM_S} \bE_P \lVert \hat{P}-P^<\rVert_1 \asymp \sqrt{\frac{S}{n\ln n}} + \tilde{\Theta}\left(\sqrt{\frac{S}{n}}\wedge n^{-\frac{1}{3}}\right).
	\end{align*}
	Furthermore, the estimator $\hat{P}$ constructed in Section \ref{sec.construction} does not require the knowledge of the support size $S$. 
\end{theorem}

The following corollary is immediate.
\begin{corollary}\label{cor.sample_complexity}
	There exists an estimator $\hat{P}$ for the unlabeled distribution $\hat{P}$ under sorted $\ell_1$ loss if and only if $n\gg \frac{S}{\ln S}$.
\end{corollary}

Corollary \ref{cor.sample_complexity} shows that as opposed to the requirement $n=\omega(S)$ in consistently estimating $P$, estimating the sorted distribution $P^<$ only requires sub-linear samples $n=\omega(\frac{S}{\ln S})$. The following corollary shows that when the support size $S$ is not too small, the empirical frequency $P_n$ which is minimax optimal for estimating $P$ is no longer optimal in estimating $P^<$.

\begin{corollary}\label{cor.MLE}
	For $n\gtrsim \frac{S}{\ln S}$, the minimax rate-optimal estimator $\hat{P}$ for $P^<$ outperforms the sorted empirical distribution $P_n^<$ if and only if $S\gg \tilde{\Theta}(n^{\frac{1}{3}})$, where $P_n$ denotes the empirical distribution.
\end{corollary}

Corollary \ref{cor.MLE} shows that the minimax rate-optimal estimator outperforms the baseline (i.e., the sorted empirical distribution) when $\frac{S}{\ln S} \lesssim n \ll \tilde{\Theta}(S^3)$. The constraint that the sample size $n$ cannot be too large is indeed natural: for larger sample size, there is not enough ambiguity between the relative magnitude of probabilities of each symbol, and the problem of estimating sorted distribution is essentially reduced to that of estimating the original unsorted distribution. Specifically, we show that the error $\Theta(\sqrt{\frac{S}{n\ln n}})$ can be achieved if and only if $\frac{S}{\ln S} \lesssim n \lesssim \tilde{\Theta}(S^3)$; in contrast, \cite{valiant2017estimating} can only achieve it under a different Wasserstein distance when $\frac{S}{\ln S}\lesssim n\lesssim S$. 

We then demonstrate the performance of plugging-in our sorted distribution estimate into certain symmetric functionals of the following form:
\begin{align}\label{eq.general_functional}
F(P) &= \sum_{i=1}^S f(p_i), \qquad f\in C[0,1], f(0)=0.
\end{align}
We will be mainly interested in the case where $f$ is non-smooth and the estimation of $F(P)$ becomes challenging. Concretely, we consider entropy $H(P)$, the power sum function $F_\alpha(P)$ and the support size $S(P)$, which are given by
\begin{align*}
H(P) &\triangleq \sum_{i=1}^S p_i\ln \frac{1}{p_i}, \qquad  P\in\calM_S, \\
F_\alpha(P) &\triangleq \sum_{i=1}^S p_i^\alpha, \qquad  P\in\calM_S, 0<\alpha<1, \\
S(P) &\triangleq \sum_{i=1}^S \mathbbm{1}(p_i\neq 0), \qquad  P\in \calD_k\triangleq \left\{P\in\calM_S: p_i\ge \frac{1}{k}, i\in [S]\right\}.
\end{align*}
The exact minimax rates for these functionals have been obtained in \cite{wu2016minimax,Jiao--Venkat--Han--Weissman2015minimax,wu2015chebyshev}, respectively. For these functionals, the plug-in approach $F(\hat{P})$ of the estimator $\hat{P}$ of the sorted distribution $P^<$ (as in Theorem \ref{thm.achievability}) achieves the corresponding minimax risk when the sample size $n$ is not too large.
\begin{theorem}[Informal]\label{thm.functional_informal}
	For $F(P)=H(P), F_\alpha(P)$ with $\alpha\in (0,1)$ and $S(P)$, the plug-in estimator $\hat{F}=F(\hat{P})$ achieves the corresponding minimax risk when $n$ is not too large. In particular, $\hat{F}$ attains the optimal sample complexity $n\gg \frac{S}{\ln S}$ and $n\gg \frac{k}{\ln k}$ to achieve a vanishing error.
\end{theorem}

We refer the explicit construction of the plug-in estimator $\hat{F}=F(\hat{P})$ and the precise statement of Theorem \ref{thm.functional_informal} to Section \ref{sec.functional}. In summary, Theorem \ref{thm.functional_informal} provides an affirmative answer to Question \ref{question} and shows that our estimator $\hat{P}$ for the sorted distribution $P^<$ is ``universal" in the sense that the plug-in approach yields a near-minimax estimator for various functionals.

The rest of this paper is organized as follows. Section \ref{sec.construction} relates the problem of estimating sorted distribution $P^<$ to the distribution estimation problem under Wasserstein distance, where the idea of local moment matching is motivated and the final estimator $\hat{P}$ is constructed. Section \ref{sec.functional} discusses the application of sorted distribution estimation to symmetric functional estimation in detail, and proves Theorem \ref{thm.functional_informal}. Theorem \ref{thm.main} is then proved via a combination of the achievability part in Section \ref{sec.analysis} and the converse part in Section \ref{sec.lower}. Some auxiliary lemmas and their proofs are deferred in the appendices.

%

\emph{Notation: } For a finite set $A$, let $|A|$ denote its cardinality; $[n]\triangleq \{1,\cdots,n\}$; lattice operations $\wedge, \vee$ are defined as $a\wedge b=\min\{a,b\}, a\vee b=\max\{a,b\}$; let $\calM_S$ denote the probability simplex over $S$ elements, and $\mathsf{Poly}_K$ be the space of all polynomials of degree at most $K$; for non-negative sequences $\{a_n\}$ and $\{b_n\}$, the notation $a_n\lesssim b_n$ (or $b_n\gtrsim a_n, a_n=O(b_n), b_n=\Omega(a_n)$) means $\limsup_{n\to\infty} \frac{a_n}{b_n}<\infty$, and $a_n\ll b_n$ ($b_n\gg a_n, a_n=o(b_n), b_n=\omega(a_n)$) means $\limsup_{n\to\infty}\frac{a_n}{b_n}=0$, and $a_n\asymp b_n$ (or $a_n=\Theta(b_n)$) is equivalent to both $a_n\lesssim b_n$ and $b_n\lesssim a_n$.

\section{Estimator Construction}\label{sec.construction}
In this section, we make use of the duality in Wasserstein distance to relate sorted distribution estimation to the estimation of Lipschitz functionals, and introduce the duality between moment matching and polynomial approximation. Based on these insights, finally we construct the estimator via local moment matching.
\subsection{Duality of Wasserstein Distance}
We first introduce the Wasserstein distance. 
\begin{definition}[Wasserstein Distance]
	Let $(S,d)$ be a separable metric space, and $P,Q$ be two Borel probability measures on $S$. The Wasserstein distance between $P,Q$ is defined as
	\begin{align*}
	W(P,Q) \triangleq \inf_{\cL(X)=P, \cL(Y)=Q} \mathbb{E}[ d(X,Y)],
	\end{align*}
	where the infimum is taken over all possible couplings between $S$-valued random variables $X,Y$ with marginals $P$ and $Q$, respectively.
\end{definition}

The key reason why we introduce the Wasserstein distance lies on the following lemma.
\begin{definition}
For any vector $P=(p_1,\cdots,p_S)$, we define $\mu_P$ to be the uniform probability measure on the multiset $\{p_1,\cdots,p_S\}$.
\end{definition}
\begin{lemma}\label{lemma.wasserstein}
	For any two vectors $P,Q$, we have
	\begin{align*}
	\lVert P^< - Q^< \rVert_1 = S\cdot W(\mu_P,\mu_Q)
	\end{align*}
	with $d(x,y)=|x-y|$ being the usual Euclidean metric.
\end{lemma}

In other words, in order to estimate the sorted distribution $P^<$ in terms of the $\ell_1$ distance, it is equivalent to finding some distribution $\hat{P}$ such that the Wasserstein distance between $\mu_P$ and $\mu_{\hat{P}}$ is small. However, $\mu_{\hat{P}}$ must be a discrete measure, which may complicate the estimator construction. Fortunately, the following randomization procedure and Lemma \ref{lemma.randomization} show that it also suffices to find any (possibly non-atomic) probability measure $\hat{\mu}_P$ over the real line such that $W(\mu_P,\hat{\mu}_P)$ is small: 
\begin{definition}[Randomized Discretization]\label{def.randomization}
	Given a support size $S$ and any probability measure $\mu$ over the real line, the following procedure outputs an $S$-dimensional vector $Q=(q_1,\cdots,q_S)$:
	\begin{enumerate}
		\item Let $F$ be the CDF of $\mu$, and $U_1,\cdots,U_S$ be $S$ independent random variables uniformly distributed on $[0,1]$;
		\item For each $i=1,\cdots,S$, define $
		q_i \triangleq F^{-1}\left(\frac{i-U_i}{S}\right) $, 
		where the inverse $F^{-1}(\cdot)$ is defined as
$
		F^{-1}(t) \triangleq \sup\{x: F(x)\le t\}.
$
		\item Finally, form the vector $Q=(q_1,\cdots,q_S)$.
	\end{enumerate}
\end{definition}
\begin{lemma}\label{lemma.randomization}
	Let $P$ be an $S$-dimensional vector, and $\mu$ be any probability measure on $\mathbb{R}$. If $Q$ is the returned vector of the previous randomization procedure, we have
	\begin{align*}
	\bE W(\mu_P,\mu_Q) = W(\mu_P,\mu)
	\end{align*}
	where the expectation is taken with respect to the randomness in the randomized procedure.
\end{lemma}

To find a suitable probability measure $\hat{\mu}$ such that the Wasserstein distance $W(\mu_P,\hat{\mu})$ is small, it will be helpful to recall the well-known dual representation of the Wasserstein distance:
\begin{lemma}\cite{kantorovich1958space}\label{lemma.dual}
	For two Borel probability measures $P,Q$ on a separable metric space $(S,d)$, the following duality result holds:
	\begin{align*}
	W(P,Q) = \sup_{f: \lVert f \rVert_{\text{\rm Lip}}\le 1} \bE_P f(X) - \bE_Q f(X)
	\end{align*}
	where $X$ is a random variable taking value in $S$ with distribution $P$ or $Q$, and the Lipschitz norm is defined as
$
	\lVert f\rVert_{\text{\rm Lip}} \triangleq \sup_{x\neq y\in S}\frac{|f(x)-f(y)|}{d(x,y)}.
$
\end{lemma}

The previous lemma shows that we need to find some $\hat{\mu}$ such that $\bE_{\hat{\mu}}f(X)$ is close to $\bE_{\mu_P}f(X)$ for any real-valued function $f$ with Lipschitz norm at most one. Moreover, by definition of $\mu_P$, we have
$
\bE_{\mu_P}f(X) = S^{-1}\sum_{i=1}^S f(p_i)
$
for any function $f$. In other words, we need to tackle the problem of \emph{functional estimation} of the form $\sum_{i=1}^S f(p_i)$ for all $1$-Lipschitz functions \emph{simultaneously}. There are two fundamental difficulties in this problem:
\begin{enumerate}
	\item Estimation of functionals is hard in general;
	\item The space of all $1$-Lipschitz functions is infinite dimensional. 
\end{enumerate} 

The next subsection will be devoted to overcoming these two difficulties.

\subsection{Duality between Moment Matching and Approximation}
In this subsection we present answers to the previous questions. The first step is to estimate the functional of the form $\sum_{i=1}^S f(p_i)$ for some fixed $1$-Lipschitz function $f$. It may be tempted to use the plug-in approach $\sum_{i=1}^S f(\hat{p}_i)$, where $\hat{p}_i$ denotes the empirical probability of the symbol $i$. Note that this approach will return the empirical distribution as the distribution estimate in the end. However, it has been shown in previous works (e.g., \cite{Jiao--Venkat--Han--Weissman2015minimax}) that \emph{bias} is the dominating error in the estimation of functionals, and the plug-in approach incurs too much bias. This observation motivates us to look for proper functions $f$ such that there exists an unbiased estimator of $\sum_{i=1}^S f(p_i)$, and the sufficient and necessary condition is that $f$ must be a polynomial of degree at most $n$. Specifically, for $n\hat{p}\sim \mathsf{B}(n,p)$, we have
\begin{align*}
\bE\left[\frac{n\hat{p}(n\hat{p}-1)\cdots(n\hat{p}-k+1)}{n(n-1)\cdots(n-k+1)}\right] = p^k, \qquad k=0,1,\cdots,n.
\end{align*}
Hence, polynomials are \emph{easy} functionals for estimation, and we may restrict the function $f$ to be polynomials up to a certain degree.

The next step is to resolve the problem that the space of all $1$-Lipschitz functions is infinite dimensional. However, if we could accurately estimate $\sum_{i=1}^S f(p_i)$ for monomials, i.e.,
$
\sum_{i=1}^S \hat{p}_i^k \approx \sum_{i=1}^S p_i^k$ for $k=0,1,\cdots,K,
$
then for any $1$-Lipschitz function $f$ and any polynomial $P$ of degree at most $K$, we have
\begin{align*}
\left|\sum_{i=1}^S (f(\hat{p}_i) - f(p_i))\right| \approx \left|\sum_{i=1}^S [(f(\hat{p}_i) - P(\hat{p}_i)) - (f(p_i)-P(p_i))] \right| \le 2S\cdot \inf_{P\in \mathsf{Poly}_K} \|f-P\|_\infty
\end{align*}
is small for large $K$. Hence, we can approximate the infinite-dimensional Lipschitz ball via a finite collection of functions, and then find an distribution to match these basis functions. When we choose monomials as the basis, we arrive at the \emph{moment matching}; the reason why monomials are chosen as the basis will be detailed in the next subsection. 

We emphasize that the final estimator requires that we conduct \emph{moment matching} locally rather than globally, as shown in the next subsection. 

\subsection{Final Estimator}
Before constructing the final estimator for $P^<$, we introduce the idea of Poissonization which has been widely used in related models. Specifically, in the Poissonized model, instead of drawing $n$ i.i.d samples $X_1,\cdots,X_n$ from $P$, we draw $N$ i.i.d samples $X_1,\cdots,X_N$ from $P$, where the number of samples $N\sim \spo(n)$ is a random variable. The reason why we work on the Poissonized model is that, the empirical counts $n\hat{p}_j\sim \spo(np_j)$ are independent under the Poissonized model. The following lemma relates the sorted $\ell_1$ errors in these two models:
\begin{lemma}\label{lemma.poissonization}
	Let $R(n,S), R_P(n,S)$ be the minimax risk under sorted $\ell_1$ loss in the Multinomial and Poissonized models, respectively. The following inequality holds:
	\begin{align*}
	\frac{1}{2}R(2n,S) \le R_P(n,S) \le R(\frac{n}{2},S) + 2\exp(-\frac{n}{8}).
	\end{align*}
\end{lemma}

By Lemma \ref{lemma.poissonization}, it suffices to focus on the Poissonized model, where the estimator $\hat{P}$ for the sorted distribution $P^<$ is constructed as follows:
\begin{enumerate}
	\item Split the samples into two parts, i.e., attach a random label uniformly distributed on $\{1,2\}$ independently to each observation $X_1,\cdots,X_N$, and the observations are partitioned into two parts $X^{(1)}, X^{(2)}$ according to the label. By the property of Poisson distribution, each subset of the samples determines a Poissonized sampling model with rate replaced by $n/2$, and different subsets are independent. In the sequel we redefine $n/2$ as $n$ for notational simplicity;
	\item For each part of the samples and $i\in [S]$, conpute the empirical frequency $\hat{p}_{i,1}, \hat{p}_{i,2}$. The empirical frequency $\hat{p}_{i,1}$ in the first part will be used to determine the local domain where the true probability $p_i$ lies, and $\hat{p}_{i,2}$ in the second part will be used for the estimation of $P^<$;
	\item Fixing a universal constant $c_1>0$, partition the unit interval $[0,1]$ into $M$ sub-intervals $I_1,I_2,\cdots,I_M$, where
$
	I_j \triangleq [\frac{c_1\ln n}{n}\cdot (j-1)^2, \frac{c_1\ln n}{n}\cdot j^2]$ for $j\in [M]
$.\footnote{Boundary points can belong to either intervals, as long as $\{I_j\}_{j=1}^M$ constitutes a legitimate partition of $[0,1]$.}
	We define $x_j\triangleq \frac{c_1j(j-1)\ln n}{n}$ as the ``center'' of the interval $I_j$. Similarly, we also define a slightly ``enlarged'' version of $I_j$:
$
	\tilde{I}_j \triangleq [\frac{c_1\ln n}{n}\cdot (j-\frac{3}{2})^2\mathbbm{1}(j\ge 2), \frac{c_1\ln n}{n}\cdot (j+1)^2]$ for $ j\in [M].
$
	Without loss of generality we assume that the number of sub-intervals $M=\sqrt{\frac{n}{c_1\ln n}}$ is an integer;
	\item Fixing universal constants $c_2,c_3>0$, in each sub-interval $\tilde{I}_j$ we solve the following convex optimization problems:
	\begin{itemize}
		\item If $j\ge 2$, check whether there exists a measure $\mu_j$ on $\tilde{I}_j$ such that
		\begin{align}
		\mu_j(\tilde{I}_j) &= S_j, \label{eq.LP_1}\\
		\left|\int_{\tilde{I}_j} (x-x_j)^k\mu_j(dx) - \sum_{i=1}^S \mathbbm{1}(\hat{p}_{i,1}\in I_j)g_{k,x_j}(\hat{p}_{i,2})\right| &\le \sqrt{S_j\ln n}\cdot \left(\frac{c_3j\ln n}{n}\right)^{k} \label{eq.LP_2}
		\end{align}
		hold simultaneously for $k=1,2,\cdots,K\triangleq c_2\ln n$, where $S_j\triangleq \sum_{i=1}^S \mathbbm{1}(\hat{p}_{i,1}\in I_j)$ is the number of symbols whose empirical probability lies in the interval $I_j$, and
		\begin{align*}
		g_{k,x}(p) &\triangleq \sum_{l=0}^k \binom{k}{l}(-x)^{k-l} \prod_{l'=0}^{l-1}\left(p-\frac{l'}{n}\right).
		\end{align*}
		If there exists a feasible solution, pick an arbitrary one; otherwise, report ``failure'' for the interval $\tilde{I}_j$;
		\item If $j=1$, solve the following minimization program over all measures $\mu_1$ on $\tilde{I}_1$:
		\begin{equation}
		\begin{aligned}
		\min & ~ \mu_1(\tilde{I}_1)  \\
		\text{\rm s.t.}     & ~ \left|\int_{\tilde{I}_1} (x-x_1)^k\mu_1(dx) - \sum_{i=1}^S \mathbbm{1}(\hat{p}_{i,1}\in I_1)g_{k,x_1}(\hat{p}_{i,2})\right| \le \sqrt{\mu_1(\tilde{I}_1)\ln n}\cdot \left(\frac{c_3\ln n}{n}\right)^{k} \\
		& ~ \qquad \text{for all }k=1,\cdots,K.
		\end{aligned}
		\label{eq.LP_min}
		\end{equation}
		Report the solution as $\mu_1$; if this problem is infeasible, report ``failure".
	\end{itemize}
	\item Construct a measure $\hat{\mu}$ on $[0,1]$ as follows: if any previous step reports ``failure'', set $\hat{\mu}$ to be an arbitrary fixed distribution on $[0,1]$; otherwise, set $\hat{\mu} = \sum_{j=1}^M \mu_j$;
	\item Finally we need to output a vector. Let $S_0\triangleq \lceil \hat{\mu}(\reals) \rceil$, we add $S_0-\hat{\mu}(\reals)\ge 0$ to the point mass $\hat{\mu}(\{0\})$, and apply the randomized discretization in Definition \ref{def.randomization} (with support size $S_0$) to transform the probability measure $\hat{\mu}/S_0$ into $\hat{P}$, which is our final estimator for $P^<$.
\end{enumerate}

A few remarks are in order:
\begin{enumerate}
	\item Choice of the partition $\{I_j\}_{j=1}^M$: the partition $\{I_j\}_{j=1}^M$ is chosen so that based on an observation $\hat{p}_{i,1}\in I_j$, the true probability mass $p_i$ can be ``localized" around $I_j$ (i.e., $p_i$ belongs to a slightly larger interval $\tilde{I}_j$) with high probability. Each interval $I_j$ coincides with the definition of ``confidence set" in \cite{Han--Jiao--Weissman2016minimaxdivergence}, and the exact meaning of ``localization" is referred to Lemma \ref{lemma.concentration}. As a result, probability masses $p_i$ in the same partition are indistinguishable, while those in different partitions can be easily distinguished. Hence, at later stages it suffices to match moments locally since this is the range of indistinguishable probability masses.
	\item Choice of $g_{k,x}(p)$: the key reason to choose $g_{k,x}(p)$ in the linear program is that in Poissonized model $n\hat{p}\sim \spo(np)$, the statistic $g_{k,x}(\hat{p})$ is an unbiased estimator of $(p-x)^k$ \cite[Example 2.8]{Withers1987}:
	\begin{align}\label{eq.unbiased}
	\bE g_{k,x}(\hat{p}) &= \sum_{l=0}^k \binom{k}{l}(-x)^{k-l}\bE\prod_{l'=0}^{l-1}\left(\hat{p}-\frac{l'}{n}\right) = \sum_{l=0}^k \binom{k}{l}(-x)^{k-l}p^l = (p-x)^k.
	\end{align}
	We will see in Section \ref{sec.functional} that estimating $P^<$ is similar to estimating symmetric functionals of $P$ where bias is the dominating factor of the error, and this fact motivates us to apply an unbiased estimator of $(p-x)^k$ in \eqref{eq.LP_2}, \eqref{eq.LP_min}.
	\item Convex optimization: the constraints in \eqref{eq.LP_1}, \eqref{eq.LP_2} are linear in the measure $\mu_j$, and thus they constitute an infinite-dimensional linear program. By squaring each sides of \eqref{eq.LP_min}, the optimization problem for $j=1$ becomes a conic quadratic programming and is thus convex. Due to its special structure, there is also a linear-programming-based way to solve \eqref{eq.LP_min}: just do bisection search for $\mu_1(\tilde{I}_1)$, and solve a linear programming to check feasibility for each $\mu_1(\tilde{I}_1)$. To overcome the infinite dimensionality, in practice we can assume that $\mu_j$ is supported on a sufficiently fine grid to obtain a finite-dimensional problem. One can also transform the feasibility program in \eqref{eq.LP_1}, \eqref{eq.LP_2} into a minimization problem, while the current form is sufficient for theoretical purposes. The idea of applying linear programming in related problems has appeared in several works, e.g., \cite{Valiant--Valiant2011,Valiant--Valiant2013estimating,valiant2015instance,kong2017spectrum,tian2017learning}.
	\item Moment matching via convex optimization: the optimization problems \eqref{eq.LP_1}, \eqref{eq.LP_2} and \eqref{eq.LP_min} are designed in such a way that the true measure $\mu_{P,j}$ (cf. \eqref{eq.mu_P_j}, which requires the knowledge of the unknown $P$) is a feasible solution with high probability (cf. Lemma \ref{lemma.feasibility}). Consequently, for any feasible solution $\mu_j$, triangle inequality ensures that the local moments of our estimator will be close to the true moments (cf. \eqref{eq.moment_match}). The degree $K$ of matched moments will be the main source of the \emph{bias} of our estimator, and the RHS of \eqref{eq.LP_2} measures the fluctuation and will become the \emph{variance}.
	\item Choice of monomials $(x-x_j)^k$: there are two reasons to choose the monomials rather than other basis functions. Firstly, there exist unbiased estimators for monomials in the Poissonized model. Secondly, the subspace spanned by monomials is an optimal basis for approximating Lipschitz functions, i.e., it attains the Kolmogorov-$n$ width of the Lipschitz ball \cite{lorentz1996constructive}.
	\item The knowledge of the support size $S$: we remark that our estimator construction is agnostic to the support size $S$. A key observation is that, although $S$ appears in the definition of $S_j$, unseen symbols will not affect $\{S_j\}_{j\ge 2}$ since $S_j$ only consists of symbols which have appeared in the first half samples for $j\ge 2$. The reason why we need a different program for $j=1$ is to deal with the unknown support size: if $S$ was known, we could simply replace \eqref{eq.LP_min} by \eqref{eq.LP_1}, \eqref{eq.LP_2} as well. The last step returns a vector $\hat{P}$ of length no longer than $S$ with high probability (cf. Corollary \ref{cor.bad_event}), and we simply fill zeros to $\hat{P}$ when evaluating $\|\hat{P}-P^<\|_1$ with $S\ge S_0$. A key observation is that, filling $m$ zeros to a vector $\hat{P}$ is equivalent to adding $m$ units to $\hat{\mu}(\{0\})$ and then applying the randomized discretization, and thus Lemmas \ref{lemma.wasserstein} and \ref{lemma.randomization} still hold without knowing the support size $S$. 
\end{enumerate}

The performance of the estimator $\hat{P}$ is summarized in the following theorem.
\begin{theorem}\label{thm.achievability}
	Let $c_1>0$ be large enough as in Lemma \ref{lemma.concentration}, $c_1>2c_2, c_3>30c_1, c_2(6\ln 2+\ln(5c_3/c_1))<\epsilon$ and $c_2\ln n\ge 1$. Then there exists a constant $C_0>0$ independent of $n,S$ such that
	\begin{align*}
	\sup_{P\in\calM_S} \bE_P \|\hat{P}-P^<\|_1 \le C_0\left(\sqrt{\frac{S}{n\ln n}} + n^{\epsilon}\left(\sqrt{\frac{S}{n}} \wedge n^{-\frac{1}{3}}\right)\right).
	\end{align*}
\end{theorem}

\section{Applications in Symmetric Functional Estimation}\label{sec.functional}
For functionals $F(\cdot)$ taking the form of \eqref{eq.general_functional}, if we define the following estimator\footnote{Our construction of $\hat{\mu}$ does not depend on $S$, so is $\hat{F}$.}
\begin{align}\label{eq.plug-in-est}
\hat{F} \triangleq \int_{\mathbb{R}} f(x)\hat{\mu}(dx) = S\cdot\int_{\mathbb{R}} f(x)\hat{\mu}^*(dx)
\end{align}
with $\hat{\mu}$ given by our estimator construction and $\hat{\mu}^*=S^{-1}\hat{\mu}$, it is straightforward to see that
\begin{align*}
|\hat{F} - F(P)| = S\cdot \left|\int_{\mathbb{R}} f(x)(\hat{\mu}^*(dx) - \mu_P(dx))\right|.
\end{align*}
If $f$ is 1-Lipschitz, since Theorem \ref{thm.achievability} guarantees that the expected Wasserstein distance $\bE_PW(\hat{\mu}^*,\mu_P)$ is small, it follows from the dual representation of Wasserstein distance (cf. Lemma \ref{lemma.dual}) that $\bE_P|\hat{F}-F(P)|$ is also small. For general non-smooth $f$, we have the following lemma:
\begin{lemma}\label{lemma.functional}
	Let the parameter configurations in Theorem \ref{thm.achievability} be fulfilled, $f(0)=0$, $\mu_{P,j}$ and $\mu_j$ be defined in \eqref{eq.mu_P_j}, \eqref{eq.LP_1}, \eqref{eq.LP_2} and \eqref{eq.LP_min}, respectively. Suppose that for each $j\in [M]$, there is a polynomial $P_j$ of degree at most $K$ with $\|f-P_j\|_{\infty,\tilde{I}_j} \le M_j\triangleq \sup_{x\neq y\in \tilde{I}_j}|f(y)-f(x)|$, and $P_1(0)=0$. For the estimator $\hat{F}$ defined in \eqref{eq.plug-in-est}, with probability at least $1-3Sn^{-4}$, the following inequality holds conditioning on the first half samples:
	\begin{align*}
	|\hat{F}-F(P)|\le C_0\sum_{j=1}^M\left(\int_{\reals} |f(x)-P_j(x)|(\mu_j(dx)+\mu_{P,j}(dx)) + n^{\epsilon}M_j\sqrt{S_j}\right),
	\end{align*}
	where $C_0>0$ is a constant independent of $n,S,f$ and $P_j$. In particular,
	\begin{align*}
	|\hat{F}-F(P)| \le C_0\sum_{j=1}^M\left(S_j\cdot\inf_{P_j\in\mathsf{Poly}_K}\|f-P_j\|_{\infty,\tilde{I}_j} + n^{\epsilon}M_j\sqrt{S_j}\right).
	\end{align*}
\end{lemma}
\begin{remark}
	The condition $\|f-P_j\|_{\infty,\tilde{I}_j} \le M_j$ is mild since a reasonable approximating polynomial should approximate $f(\cdot)$ better than a constant function on $\tilde{I}_j$.
\end{remark}

Lemma \ref{lemma.functional} shows that the plug-in approach of $\hat{F}$ does polynomial approximation \emph{implicitly} and \emph{everywhere}. Specifically, the construction of $\hat{F}$ has nothing to do with polynomial approximation, while polynomial approximation emerges in the performance analysis of $\hat{F}$. Moreover, as opposed to the traditional approximation-based estimators where an explicit and functional-dependent polynomial is required, the plug-in estimator $\hat{F}$ can essentially approximate $f$ using \emph{any} polynomials. This property is desirable, since in general $P_j$ may not be the best approximating polynomial and may be hard to design explicitly: we refer to \cite{jiao2017minimax} for such an example. Also, as opposed to the approximation-based estimators which split $\text{dom}(f)$ into ``non-smooth" and ``smooth" regimes, the plug-in estimator $\hat{F}$ does polynomial approximation everywhere. This property prevents $\hat{F}$ from achieving the optimal variance, but this is the price we need to pay to achieve a unified methodology without the dependence on $f$.

The following theorem characterizes the performance of $\hat{F}$ for $F(P)$ with $F=H,F_\alpha,S$:\footnote{For the support size functional $S(P)$, due to the additional constraint $p_i\ge \frac{1}{k}$ on the parameter set, an additional linear constraint $\mu_j((0,\frac{1}{k}))=0$ should be imposed in addition to \eqref{eq.LP_1} and \eqref{eq.LP_2}.}
\begin{theorem}\label{thm.functional}
	The plug-in estimator in \eqref{eq.plug-in-est} with $F=H, F_\alpha, S$ satisfies:
	\begin{align*}
	\sup_{P\in\calM_S} \bE_P|\hat{H}-H(P)| &\le C_0\left(\frac{S}{n\ln n} + n^{\epsilon}\left(\sqrt{\frac{S}{n}}\wedge n^{-\frac{1}{3}}\right)\right), \\
	\sup_{P\in\calM_S} \bE_P|\hat{F}_\alpha-F_\alpha(P)| &\le C_0\left(\frac{S}{(n\ln n)^\alpha} + n^{\epsilon}\left(\sqrt{\frac{S^{3-2\alpha}}{n}}\wedge n^{-\frac{\alpha}{3}}\right)\right),\\
	\sup_{P\in\calD_k} \bE_P|\hat{S}-S(P)| &\le C_0k\left(\exp\left(-\Theta\left(\sqrt{\frac{n\ln k}{k}}\right)\right) + \frac{n^{\epsilon}}{\sqrt{k}}\right),\qquad n\lesssim k\ln k,
	\end{align*}
	where $C_0>0$ is a constant independent of $n,S,k$.
\end{theorem}
\begin{remark}
	An additional condition $n\lesssim k\ln k$ is required for the support size functional $S(P)$: if $n\gg k\ln k$, the minimax risk decays super-polynomially in $n$, which makes the $O(Sn^{-4})$ failure probability in Lemma \ref{lemma.functional} become non-negligible.
\end{remark}
Compared with the minimax rates of these functionals in \cite{wu2016minimax,Jiao--Venkat--Han--Weissman2015minimax,wu2015chebyshev}, the general plug-in approach in \eqref{eq.plug-in-est} achieves the optimal total bias term, which is the leading term when $S$ or $k$ is large. As a result, the plug-in approach attains the optimal sample complexity for all these functionals, establishing Theorem \ref{thm.functional_informal}.

However, a comparison of Theorem \ref{thm.functional} and the minimax rates shows that the variance term of $\hat{F}$ is not optimal, conforming to the aforementioned intuition that everywhere polynomial approximation may incur a too large variance. Hence, among the functionals considered in Theorem \ref{thm.functional}, the general plug-in approach in \eqref{eq.plug-in-est} attains the optimal bias and thus the optimal sample complexity, but need to pay a price on the variance.

\clearpage
\bibliography{di}

\clearpage
\appendix

\section{Estimator Analysis}\label{sec.analysis}
\subsection{Controlling ``Bad Events''}
There are several types of bad events in the construction of our estimator: 
\begin{enumerate}
	\item For some symbol $i=1,\cdots,S$ and $j=1,\cdots,M$, it may happen that $p_i\notin \tilde{I}_j$ but $\hat{p}_{i,1}\in I_j$;
	\item For some $j=1,\cdots,M$, it may happen that the linear programming in \eqref{eq.LP_1}, \eqref{eq.LP_2}, \eqref{eq.LP_min} does not have a solution;
	\item In the last step, it may happen that $\hat{\mu}(\reals)>S$.
\end{enumerate}
In this subsection we show that the probability that any of these bad events occurs is negligible. The following lemma follows directly from the Poisson tail inequalities (cf. Lemma \ref{lemma.poissontail}).
\begin{lemma}\label{lemma.concentration}
	Let $c_1>0$ be large enough, then for any $i=1,2,\cdots,S$ and $j=1,2,\cdots,M$,
	\begin{align*}
	\mathbb{P}(\hat{p}_{i,1}\in I_j|p_i\notin \tilde{I}_j) \le n^{-5}.
	\end{align*}
\end{lemma} 

Based on Lemma \ref{lemma.concentration} and the union bound, we see that the first-type bad events occurs with a negligible probability. To upper bound the probability of other bad events, we need to come up with a solution $\mu_j$ which fulfills \eqref{eq.LP_1}, \eqref{eq.LP_2}, \eqref{eq.LP_min} with high probability. In the sequel we condition on a specific realization of the first half samples, and define the set of symbols falling in $I_j$ as
\begin{align*}
A_j \triangleq \{i\in\{1,2,\cdots,S\}: \hat{p}_{i,1}\in I_j \}.
\end{align*}
Note that $A_j$ is a random set depending only on the first half samples, and $|A_j|=S_j$. Moreover, $\{A_j,S_j\}_{j\ge 2}$ are uniquely determined by the first half samples, while $\{A_1,S_1\}$ may be unknown due to the unknown support size $S$ and possibly unseen symbols. Now the key observation is that, the following measure
\begin{align}\label{eq.mu_P_j}
\mu_{P,j}(\cdot) \triangleq \sum_{i\in A_j} \mathbbm{1}(p_i\in \cdot), \qquad j=1,2,\cdots,M
\end{align}
which requires the knowledge of the unknown $P$ satisfies \eqref{eq.LP_1}, \eqref{eq.LP_2}, \eqref{eq.LP_min} with high probability. Obviously, if $p_i\in \tilde{I}_j$ for any $i\in A_j$ and $j\ge 2$, the measure $\mu_{P,j}$ will be supported on $\tilde{I}_j$ and thus \eqref{eq.LP_1} holds. The following lemma shows that, given the same assumption, the measure $\mu_{P,j}$ also satisfies \eqref{eq.LP_2} and \eqref{eq.LP_min} with high probability.
\begin{lemma}\label{lemma.feasibility}
	Let $c_1>0$ be large enough as in Lemma \ref{lemma.concentration}, and $c_1>2c_2, c_3>30c_1, c_2\ln n\ge 1$. Further assume that $p_i\in \tilde{I}_j$ for any $i\in A_j$ and $j=1,2,\cdots,M$. Then conditioning on the first half samples, for $k=1,2,\cdots,K=c_2\ln n$ we have
	\begin{align*}
	\mathbb{P}\left(\left|\int_{\tilde{I}_j} (x-x_j)^k\mu_j(dx) - \sum_{i=1}^S \mathbbm{1}(\hat{p}_{i,1}\in I_j)g_{k,x_j}(\hat{p}_{i,2})\right| > \sqrt{S_j\ln n}\cdot \left(\frac{c_3j\ln n}{n}\right)^{k} \right)\le 2n^{-4}.
	\end{align*}
\end{lemma}

Lemma \ref{lemma.feasibility} shows that $\mu_{P,j}$ is a feasible solution to \eqref{eq.LP_2} with high probability for $j\ge 2$. For $j=1$, note that $\mu_{P,1}(\tilde{I}_1)=S_1$ conditioning on the events in Lemma \ref{lemma.concentration}, the measure $\hat{\mu}_{P,1}$ is also a feasible solution to \eqref{eq.LP_min}. Moreover, in this case the returned solution $\mu_1$ of \eqref{eq.LP_min} satisfies $\mu_1(\reals)\le S_1$, thus we must have $\hat{\mu}(\reals)\le S$ in the last step. Hence, based on Lemma \ref{lemma.concentration} and \ref{lemma.feasibility}, by the union bound we have the following corollary:
\begin{corollary}\label{cor.bad_event}
	Let $E$ be the event that any of the aforementioned bad events happens, then
	\begin{align*}
	\bP(E) \le Sn^{-5} + 2Mn^{-4}.
	\end{align*}
\end{corollary}
By Corollary \ref{cor.bad_event}, the probability that any bad event happens is negligible, and it thus suffices to focus on the ``good" events to analyze the performance of our estimator, which will be the focus of the next subsection.

\subsection{Local Performance of Moment Matching}
By Lemma \ref{lemma.feasibility}, the measure $\mu_{P,j}$ using the unknown true knowledge of $P$ is a feasible solution to \eqref{eq.LP_1},  \eqref{eq.LP_2} and \eqref{eq.LP_min} with high probability, and our estimator returns a perfect answer if this solution is chosen among all feasible solutions. In this subsection, we show that any feasible solution $\mu_j$ is in fact close to the true measure $\mu_{P,j}$ in terms of the Wasserstein distance. By the dual representation of Wasserstein distance (cf. Lemma \ref{lemma.dual}), we can fix any 1-Lipschitz function $f$ on $\mathbb{R}$ and prove the following lemma:
\begin{lemma}\label{lemma.lmm}
	For any $j=1,\cdots,M$, let $\mu_j$ be any feasible solution to \eqref{eq.LP_1}, \eqref{eq.LP_2} or \eqref{eq.LP_min}, the true measure $\mu_{P,j}$ be given in \eqref{eq.mu_P_j}, and $c_2$ be small enough such that $c_2(6\ln 2+\ln(5c_3/c_1))<\epsilon$ with $c_2\ln n\ge 1$. Assuming all good events happen and conditioning on the first half samples, for any $1$-Lipschitz function $f$ on $\mathbb{R}$ with $f(0)=0$,
	\begin{align*}
	\left|\int_{\tilde{I}_j} f(x)\mu_j(dx) - \int_{\tilde{I}_j} f(x)\mu_{P,j}(dx)\right| \le C_0\left(\int_{\mathbb{R}} \sqrt{\frac{x}{n\ln n}}(\mu_j(dx)+\mu_{P,j}(dx)) + \frac{j\sqrt{S_j}}{n^{1-\epsilon}}\right)
	\end{align*}
	where $C_0$ is a constant independent of $n,S$ and $f$.
\end{lemma}
\begin{remark}
	The condition $f(0)=0$ is important for $j=1$: \eqref{eq.LP_min} cannot ensure that $\mu_1(\mathbb{R})=\mu_{P,1}(\mathbb{R})$ since it does not contain the total mass constraint in \eqref{eq.LP_1}.
\end{remark}

The remainder of this subsection is devoted to the proof of Lemma \ref{lemma.lmm}. For $j\ge 2$, by assumption both $\mu_{P,j}$ and $\mu_j$ are feasible solutions to \eqref{eq.LP_1} and \eqref{eq.LP_2}, i.e., they are supported on $\tilde{I}_j$ with the same total mass $S_j$, and by triangle inequality we have
\begin{align}\label{eq.moment_match}
\left|\int_{\tilde{I}_j} (x-x_j)^k\mu_j(dx) - \int_{\tilde{I}_j} (x-x_j)^k\mu_{P,j}(dx)\right| \le 2\sqrt{S_j\ln n}\cdot\left(\frac{c_3j\ln n}{n}\right)^k
\end{align}
for any $k=1,2,\cdots,K=c_2\ln n$. As a result, fixing any polynomial
\begin{align*}
P(x) = \sum_{k=0}^K a_k(x-x_j)^k
\end{align*}
on $\tilde{I}_j$, triangle inequality together with \eqref{eq.moment_match} gives
\begin{align}
&\left|\int_{\tilde{I}_j} f(x)(\mu_j(dx)-\mu_{P,j}(dx))\right| \nonumber\\
&\le \left|\int_{\tilde{I}_j} (f(x)-P(x))(\mu_j(dx)-\mu_{P,j}(dx))\right| + \left|\int_{\tilde{I}_j} P(x)(\mu_j(dx)-\mu_{P,j}(dx))\right| \nonumber\\
&\le \underbrace{\int_{\tilde{I}_j} |f(x)-P(x)|(\mu_j(dx) + \mu_{P,j}(dx))}_{\triangleq B_1} + \underbrace{\sum_{k=1}^K |a_k|\cdot 2\sqrt{S_j\ln n}\left(\frac{c_3j\ln n}{n}\right)^k}_{\triangleq B_2} \label{eq.decompose}.
\end{align}
It's straightforward to see that \eqref{eq.moment_match}, \eqref{eq.decompose} also hold for $j=1$, while in \eqref{eq.decompose} we need to add an additional assumption that the constant term of $P(x)$ is zero. 

The inequality \eqref{eq.decompose} holds for any polynomial $P(\cdot)$ of degree at most $K$, and both terms $B_1$ and $B_2$ depend on the choice of $P$. We shall choose $P$ to be the best approximating polynomial of $f(x)$ on $\tilde{I}_j$ in the uniform norm, i.e.,
\begin{align*}
P(x) \triangleq \arg\min_{Q\in \mathsf{Poly}_K} \max_{x\in \tilde{I}_j} |Q(x) - f(x)|.
\end{align*}
It is easy to see that this choice of $P$ will result in a small value of $B_1$, while we need the following well-known Jackson's inequality in approximation theory to upper bound $B_1$ quantitatively:
\begin{lemma}\cite{devore1976degree}
	Let $K>0$ be any integer, and $[a,b]\subset \mathbb{R}$ be any bounded interval. For any $1$-Lipschitz function $f$ on $[a,b]$, there exists a universal constant $C$ independent of $K,f$ such that there exists a polynomial $P(\cdot)$ of degree at most $K$ such that
	\begin{align}\label{eq.approx_pointwise}
	|f(x) - P(x)| \le \frac{C\sqrt{(b-a)(x-a)}}{K}, \qquad \forall x\in [a,b].
	\end{align}
	In particular, the following norm bound holds:
	\begin{align}\label{eq.approx_norm}
	\sup_{x\in [a,b]} |f(x)-P(x)| \le \frac{C(b-a)}{K}.
	\end{align}
\end{lemma}

We use the pointwise bound \eqref{eq.approx_pointwise} and the norm bound \eqref{eq.approx_norm} to upper bound $B_1$ for the case $j=1$ and $j\ge 2$, respectively. If $j=1$, we have $\tilde{I}_j=[0,\frac{9c_1\ln n}{4n}]$, then \eqref{eq.approx_pointwise} with $x=0$ and $f(0)=0$ gives $P(0)=0$, thus \eqref{eq.decompose} holds for $P$. Moreover,
\begin{align}
B_1 &\le \frac{C}{K}\int_{\tilde{I}_1} \sqrt{\frac{9c_1x\ln n}{4n}} (\mu_1(dx)+\mu_{P,1}(dx)) \nonumber\\
&= \frac{C}{c_2}\sqrt{\frac{9c_1}{4}}\cdot \int_{\tilde{I}_1} \sqrt{\frac{x}{n\ln n}}(\mu_1(dx)+\mu_{P,1}(dx)). \label{eq.B_1_small}
\end{align}
If $j\ge 2$, recall that $\tilde{I}_j=[\frac{c_1(j-3/2)^2\ln n}{n},\frac{c_1(j+1/2)^2\ln n}{n}]$, the norm bound \eqref{eq.approx_norm} gives
\begin{align}
B_1 &\le \frac{C}{K}\left(\frac{c_1(j+1/2)^2\ln n}{n}-\frac{c_1(j-3/2)^2\ln n}{n}\right)\cdot \int_{\tilde{I}_j} (\mu_j(dx)+\mu_{P,j}(dx)) \nonumber\\
&\le \frac{C}{K}\left(\frac{c_1(j+1/2)^2\ln n}{n}-\frac{c_1(j-3/2)^2\ln n}{n}\right)\cdot \int_{\tilde{I}_j} \sqrt{\frac{x}{\frac{c_1(j-3/2)^2\ln n}{n}}}(\mu_j(dx)+\mu_{P,j}(dx)) \nonumber\\
&\le \frac{12C\sqrt{c_1}}{c_2} \cdot \int_{\tilde{I}_j} \sqrt{\frac{x}{n\ln n}}(\mu_j(dx)+\mu_{P,j}(dx)). \label{eq.B_1_large}
\end{align}
A combination of \eqref{eq.B_1_small} and \eqref{eq.B_1_large} gives that for any $j\in [M]$,
\begin{align}\label{eq.B_1_bound}
B_1 \le \frac{12C\sqrt{c_1}}{c_2} \cdot \int_{\mathbb{R}} \sqrt{\frac{x}{n\ln n}}(\mu_j(dx)+\mu_{P,j}(dx)).
\end{align}

To upper bound $B_2$, we need to obtain upper bounds on the coefficients $|a_k|$ of the best approximating polynomial. We invoke Lemma \ref{lem.polycoeff} here: the polynomial $P_0(x)=P(x)-f(x_j)$ defined on $\tilde{I}_j$ satisfies
\begin{align*}
|P_0(x)| &\le |P(x)-f(x)| + |f(x)-f(x_j)| \\
&\le \frac{C}{K}\left(\frac{c_1(j+1)^2\ln n}{n}-\frac{c_1(j-3/2)^2\ln n}{n}\right) + \left(\frac{c_1(j+1)^2\ln n}{n} - \frac{c_1j(j-1)\ln n}{n}\right) \\
&\le \left(1+\frac{C}{K}\right)\frac{5c_1j\ln n}{n}.
\end{align*}
Now applying Lemma \ref{lem.polycoeff} with $A=\left(1+\frac{C}{K}\right)\frac{5c_1j\ln n}{n}$, $[a,b]=[-\frac{c_1((2j-9/4)\vee 0)\ln n}{n},\frac{c_1(3j+1)\ln n}{n}]$, for any $k=1,2,\cdots,K$ we have
\begin{align*}
|a_k| &\le 2^{\frac{7K}{2}}\left(1+\frac{C}{K}\right)\frac{5c_1j\ln n}{n}\cdot \left(\frac{c_1j\ln n}{5n}\right)^{-k}(5^K+1) \\
&\le 25\left(1+\frac{C}{K}\right)2^{6K}\cdot \left(\frac{c_1j\ln n}{5n}\right)^{1-k}.
\end{align*}
Hence, the quantity $B_2$ can be upper bounded as
\begin{align}
|B_2| &\le \sum_{k=1}^K 25\left(1+\frac{C}{K}\right)2^{6K}\cdot \left(\frac{c_1j\ln n}{5n}\right)^{1-k}\cdot 2\sqrt{S_j\ln n}\left(\frac{c_3j\ln n}{n}\right)^k \nonumber\\
&\le 10c_1(1+C)n^{c_2(6\ln 2+\ln(5c_3/c_1))}(\ln n)^{\frac{3}{2}}\cdot \frac{j\sqrt{S_j}}{n}. \label{eq.B_2_bound}
\end{align}

Now a combination of \eqref{eq.B_1_bound} and \eqref{eq.B_2_bound} completes the proof of Lemma \ref{lemma.lmm}.

\subsection{Overall Performance}
In this section we are about to establish Theorem \ref{thm.achievability}. Note that finally the true measure $\mu_P$ and the measure $\hat{\mu}^*$ as the input of randomized discretization are given by (conditioning on no failures)
\begin{align*}
\mu_P &= \frac{1}{S}\sum_{i=1}^M \mu_{P,j}, \\
\hat{\mu}^* &= \frac{1}{S}\sum_{i=1}^M \mu_j + \left(1-\frac{\hat{\mu}(\reals)}{S}\right)\delta_0
\end{align*}
where $\delta_0(\cdot)$ is the Dirac delta point mass at zero. By Lemma \ref{lemma.wasserstein} and Lemma \ref{lemma.randomization}, the sorted $\ell_1$ risk of $\hat{P}$ satisfies
\begin{align*}
\bE_P \|\hat{P}-P^<\|_1 = S\cdot\bE_P W(\mu_P,\hat{\mu}^*).
\end{align*}
Using the dual representation of the Wasserstein distance (cf. Lemma \ref{lemma.dual}), we further have
\begin{align}
S\cdot\bE_P W(\mu_P,\hat{\mu}) &= S\cdot \bE_P \sup_{f: \|f\|_{\text{Lip}}\le 1} \int_{\mathbb{R}} f(x)(\mu_P(dx) - \hat{\mu}^*(dx))\nonumber \\
&\stepa = S\cdot \bE_P \sup_{f: \|f\|_{\text{Lip}}\le 1, f(0)=0} \int_{\mathbb{R}} f(x)(\mu_P(dx) - \hat{\mu}^*(dx))\nonumber \\
&= \bE_P \sup_{f: \|f\|_{\text{Lip}}\le 1, f(0)=0} \sum_{j=1}^M \int_{\mathbb{R}} f(x)(\mu_{P,j}(dx) - \mu_j(dx)) \label{eq.final_risk}
\end{align}
where (a) follows from $\mu_P(\reals)=\hat{\mu}^*(\reals)=1$.

Suppose that the condition of Lemma \ref{lemma.lmm} holds, then each summand admits the following ``bias--variance" decomposition:
\begin{align}\label{eq.bias_var}
\int_{\mathbb{R}} f(x)(\mu_{P,j}(dx) - \mu_j(dx)) \le C_0\left(\int_{\mathbb{R}} \sqrt{\frac{x}{n\ln n}}(\mu_j(dx)+\mu_{P,j}(dx)) +  \frac{j\sqrt{S_j}}{n^{1-\epsilon}}\right).
\end{align}
The first term in \eqref{eq.bias_var} corresponds to the ``bias", which is the remaining error even after the first $K$ moments are exactly matched. The second term in \eqref{eq.bias_var} corresponds to the ``variance", which is caused by the imperfect moment matching. Since $C_0$ is independent of $f$, we have
\begin{align}
&\sup_{f: \|f\|_{\text{Lip}}\le 1, f(0)=0} \sum_{j=1}^M \int_{\mathbb{R}} f(x)(\mu_{P,j}(dx) - \mu_j(dx))\nonumber\\
& \qquad\qquad\le C_0\left(\underbrace{\int_{\mathbb{R}} \sqrt{\frac{x}{n\ln n}}(\hat{\mu}(dx)+S\mu_{P}(dx))}_{\triangleq \text{``total bias" } B} + \underbrace{\sum_{j=1}^M \frac{j\sqrt{S_j}}{n^{1-\epsilon}}}_{\triangleq \text{``total variance" } V}\right).\label{eq.total_bias_var}
\end{align}

We first upper bound the total variance $V$. Using the fact that $\sum_{i=1}^S p_i=1$, and $p_i\ge \frac{c_1j^2\ln n}{16n}$ for any $i\in A_j, j\ge 2$, we have
\begin{align}\label{eq.total_mass}
1 = \sum_{i=1}^S p_i \ge \sum_{j=2}^M \sum_{i\in A_j} p_i \ge \frac{c_1\ln n}{16n} \sum_{j=2}^M \sum_{i\in A_j}j^2 = \frac{c_1\ln n}{16n} \sum_{j=2}^M j^2S_j.
\end{align}
Moreover, $S_1\le S$, thus by defining $J=\{j\in [M]: S_j\neq 0\}$, we have
\begin{align*}
\sum_{j=1}^M j\sqrt{S_j} &\le \sqrt{S} + \sum_{j=2}^M j\sqrt{S_j} \le \sqrt{S} + \sum_{j\in J} j\sqrt{S_j} = \sqrt{S}+ \left(|J|\cdot \sum_{j\in J}j^2S_j\right)^{\frac{1}{2}} \le \sqrt{S} + 4\sqrt{\frac{{n|J|}}{{c_1\ln n}}}. 
\end{align*}
Now we obtain upper bounds for $|J|$, i.e., the number of sub-intervals in the partition which contains any symbol in the first half samples. Trivially, $|J|\le S$, and \eqref{eq.total_mass} gives
\begin{align*}
1 \ge \frac{c_1\ln n}{16n} \sum_{j=2}^M j^2S_j \ge \frac{c_1\ln n}{16n} \sum_{j\in J-\{1\}} j^2 \ge \frac{c_1\ln n}{16n}\sum_{j=2}^{|J|}j^2 \ge \frac{c_1\ln n}{48n} (|J|-1)^3
\end{align*}
implying that $|J|\le \left(\frac{48n}{c_1\ln n}\right)^{\frac{1}{3}}+1$. As a result, we conclude that
\begin{align}\label{eq.var_sum_bound}
\sum_{j=1}^M j\sqrt{S_j} \le \sqrt{S} + 4\sqrt{\frac{n}{c_1\ln n}} \left(\sqrt{S} \wedge \sqrt{\left(\frac{48n}{c_1\ln n}\right)^{\frac{1}{3}}+1} \right)
\end{align}
and consequently
\begin{align}\label{eq.total_var}
V \lesssim n^{\epsilon}\left(\sqrt{\frac{S}{n}} \wedge n^{-\frac{1}{3}}\right).
\end{align}

Now we upper bound the total bias. By definition of $\mu_P$, we know that
\begin{align*}
\int_{\reals} x\cdot S\mu_P(dx) = \sum_{i=1}^S p_i = 1.
\end{align*}
Moreover, summing over $j=1,\cdots,M$ in \eqref{eq.moment_match} for $k=1$ gives
\begin{align*}
\left|\int_{\mathbb{R}} x(S\mu_P(dx) - \hat{\mu}(dx)) \right| \le \sum_{j=1}^M 2\sqrt{S_j\ln n} \cdot\frac{c_3j\ln n}{n} = 2c_3(\ln n)^{\frac{3}{2}}\cdot \sum_{j=1}^M \frac{j\sqrt{S_j}}{n}.
\end{align*}
Hence, by \eqref{eq.var_sum_bound} and the triangle inequality, we also have $\int_{\mathbb{R}} x\hat{\mu}(dx) = 1+o(1)$. As a result, by Cauchy--Schwartz the total bias can be upper bounded as
\begin{align}\label{eq.total_bias}
B \le \frac{1}{\sqrt{n\ln n}}\left(S\cdot\sqrt{\int_{\reals} x\mu_P(dx)} + \sqrt{S\cdot\int_{\reals} x\hat{\mu}(dx)}\right) \lesssim \sqrt{\frac{S}{n\ln n}}.
\end{align}

Finally, a combination of \eqref{eq.final_risk}, \eqref{eq.total_bias_var}, \eqref{eq.total_var}, \eqref{eq.total_bias}, Corollary \ref{cor.bad_event} and the fact that $\|\hat{P}-P^<\|_1\le 2$  gives that
\begin{align*}
\sup_{P\in \calM_S} \bE_P \|\hat{P}-P^<\|_1 &\lesssim \sqrt{\frac{S}{n\ln n}} + n^{\epsilon}\left(\sqrt{\frac{S}{n}}\wedge n^{-\frac{1}{3}}\right)  + 2\cdot(2Mn^{-4}+Sn^{-5}) \\
&\lesssim \sqrt{\frac{S}{n\ln n}} + n^{\epsilon}\cdot \left(\sqrt{\frac{S}{n}}\wedge n^{-\frac{1}{3}}\right)
\end{align*}
which completes the proof of Theorem \ref{thm.achievability}.

\section{Minimax Lower Bound}\label{sec.lower}
In this section we establish the following lower bound:
\begin{theorem}\label{thm.lower}
	For $n\gtrsim \frac{S}{\ln S}$, there exists a constant $c_0>0$ independent of $n,S$ such that
	\begin{align*}
	\inf_{\hat{P}}\sup_{P\in\calM_S} \bE_P\|\hat{P}-P^<\|_1 \ge c_0\left(\sqrt{\frac{S}{n\ln n}} + \left(\sqrt{\frac{S}{n}}\wedge n^{-\frac{1}{3}}\right)\right).
	\end{align*}
\end{theorem}

Notice that a combination of Theorem \ref{thm.achievability} and Theorem \ref{thm.lower} completes the proof of Theorem \ref{thm.main}. As in the proof of achievability, we call the first term in Theorem \ref{thm.lower} as ``bias" and the second term as ``variance": the techniques used to lower bound these terms mimic those which have been widely used to lower bound the bias and the variance, respectively. The next two subsections are devoted to the proof of Theorem \ref{thm.lower}.

\subsection{Lower Bound on the ``Bias"}
To prove the $\Omega(\sqrt{\frac{S}{n\ln n}})$ lower bound, we use the following ``double duality" arguments:
\begin{enumerate}
	\item Use the dual representation of Wasserstein distance (cf. Lemma \ref{lemma.dual}) to transform into estimation of Lipschitz functionals;
	\item Use the duality between moment matching and best polynomial approximation (cf. Lemma \ref{lem.measure}) to construct two measures used in the generalized Le Cam's method (cf. Lemma \ref{lemma.tsybakov}).
\end{enumerate}
Note that both these dualities are also used in the proof of the achievability part (cf. Theorem \ref{thm.achievability}), our arguments for the achievability and lower bound are in fact dual to each other.

We first make use of the first duality. Assume by contradiction that there exists an estimator $\hat{P}$ such that $\sup_{P\in\calM_S} \bE_P\|\hat{P}-P^<\|_1\ll \sqrt{\frac{S}{n\ln n}}$, then for any $1$-Lipschitz function $f(\cdot)$ on $\mathbb{R}$ and the symmetric functional $F(P)$ of the form
\begin{align*}
F(P) \triangleq \sum_{i=1}^S f(p_i),
\end{align*}
a combination of Lemma \ref{lemma.wasserstein} and \ref{lemma.dual} implies that for the estimator $F(\hat{P})=\sum_{i=1}^S f(\hat{p}_i)$, we have
\begin{align*}
\bE_P|F(\hat{P}) - F(P)| \le S\cdot \bE_PW(\mu_P,\mu_{\hat{P}}) = \bE_P\|\hat{P}-P^<\|_1 \ll \sqrt{\frac{S}{n\ln n}}.
\end{align*}
Hence, the existence of such an estimator $\hat{P}$ implies that, for any $1$-Lipschitz function $f(\cdot)$ and the corresponding symmetric functional $F(\cdot)$, we have
\begin{align}\label{eq.functional_estimation}
\inf_{\hat{F}} \sup_{P\in\calM_S} \bE_P |\hat{F}-F(P)| \ll \sqrt{\frac{S}{n\ln n}}.
\end{align}
In other words, the estimation of Lipschitz functionals are \emph{easier} than the estimation of the underlying distribution up to permutation. As a result, if we could prove that \eqref{eq.functional_estimation} breaks down for some Lipschitz functional $F(P)$, we would arrive at the desired contradiction.

Next we step into the second duality, which requires the following generalized Le Cam's method (also known as the method of two fuzzy hypotheses \cite{Tsybakov2008}). The application of this method has appeared in several works in functional estimation \cite{Lepski--Nemirovski--Spokoiny1999estimation,Cai--Low2011,Jiao--Venkat--Han--Weissman2015minimax,wu2016minimax,Han--Jiao--Weissman2016minimaxdivergence,jiao2017minimax,Han--Jiao--Mukherjee--Weissman2017adaptive,han-jiao-weissman-wu2017minimax} to deal with the bias, which motivates us to call the first term in Theorem \ref{thm.lower} as the ``bias". Given a collection of distributions $\{P_\theta: \theta\in\Theta'\}$, suppose the observation ${\bf Z}$ is distributed as $P_\theta$ with $\theta \in \Theta\subset \Theta'$. 
Let $\hat{T} = \hat{T}({\bf Z})$ be an arbitrary estimator of a function $T(\theta)$ based on $\bf Z$. Denote the total variation distance between two probability measures $P,Q$ by 
\begin{equation*}
V(P,Q) \triangleq \sup_{A\in \mathcal{A}} | P(A) - Q(A) | = \frac{1}{2} \int |p-q| d\nu,
\end{equation*}
where $p = \frac{dP}{d\nu}, q = \frac{dQ}{d\nu}$, and $\nu$ is a dominating measure so that $P \ll \nu, Q \ll \nu$.
The following general minimax lower bound follows from the same proof as \cite[Theorem 2.15]{Tsybakov2008}:
\begin{lemma}\label{lemma.tsybakov}
	Let $\sigma_0$ and $\sigma_1$ be two prior distributions on $\Theta'$.
	Suppose there exist $\zeta\in \mathbb{R}, s>0, 0\leq \beta_0,\beta_1 <1$ such that
	\begin{align*}
	\sigma_0(\theta \in \Theta: T(\theta) \leq \zeta -s) & \geq 1-\beta_0, \\
	\sigma_1(\theta \in \Theta: T(\theta) \geq \zeta + s) & \geq 1-\beta_1.
	\end{align*}
	Then
	\begin{equation*}
	\inf_{\hat{T}} \sup_{\theta \in \Theta} \bP_\theta\left( |\hat{T} - T(\theta)| \geq s \right) \geq \frac{1-V(F_1,F_0)  - \beta_0 - \beta_1}{2},
	\end{equation*}
	where 
	$F_i=\int P_\theta \sigma_i(d \theta)$ is the marginal distribution of $\mathbf{Z}$ under the prior $\sigma_i$ for $i = 0,1$, respectively.
\end{lemma}

In our application, we will set $\theta=P$, $T(\theta)=F(P)$, and $\sigma_i=\nu_i^{\otimes S}$ for $i=0,1$, where $\nu_0,\nu_1$ are priors on $[\frac{1}{S}-\sqrt{\frac{c\ln n}{nS}}, \frac{1}{S}+\sqrt{\frac{c\ln n}{nS}}]$ with some constant $c>0$. The priors $\nu_0,\nu_1$ are chosen to be the solutions of the optimization program \eqref{eq:Rstar}, whose optimal objective value is the best polynomial approximation error as shown in the following lemma.

\begin{lemma}\label{lem.measure}
	Given a compact interval $I=[a,b]$ with $a>0$, an integer $K>0$ and a continuous function $f$ on $I$, let 
	\[
	E_{K}(f;I) \triangleq \inf_{\{a_i\}} \sup_{x\in I} \left| \sum_{i=0}^K a_i x^i
	-f(x)\right| 
	\]
	denote the best uniform approximation error of $f$ by polynomials spanned by $\{1,x,\cdots,x^K\}$.
	Then
	\begin{equation}
	\begin{aligned}
	2 E_{K}(f;I) = \max & ~ \int f(t) \nu_1(dt) - \int f(t) \nu_0(dt)   \\
	\text{\rm s.t.}     & ~ \int t^{l} \nu_1(dt) = \int t^{l} \nu_0(dt), \quad l=0,\cdots,K
	\end{aligned}
	\label{eq:Rstar}
	\end{equation}
	where the maximum is taken over pairs of probability measures $\nu_0$ and $\nu_1$ supported on $I$.
\end{lemma}
Lemma \ref{lem.measure} establishes the duality between moment matching and best polynomial approximation, where moment matching helps to obtain a small total variation distance $V(F_1,F_0)$ in Lemma \ref{lemma.tsybakov}, and best polynomial approximation error gives the value $s$ in Lemma \ref{lemma.tsybakov}. Moreover, here moment matching is also done locally, for $\mathsf{supp}(\nu_i)=[\frac{1}{S}-\sqrt{\frac{c\ln n}{nS}}, \frac{1}{S}+\sqrt{\frac{c\ln n}{nS}}], i=0,1$ takes the form of a local interval as $\{I_j\}_{j=1}^M$. 

Next we specify the choice of the 1-Lipschitz function: $f(x)=|x-\frac{1}{S}|-\frac{1}{S}$, which was studied in \cite{jiao2017minimax}. In particular, \cite{jiao2017minimax} shows that with these priors and $K\asymp \ln n$ in Lemma \ref{lem.measure}, we have $s\gtrsim \sqrt{\frac{S}{n\ln n}}$ and $\beta_0, \beta_1, V(F_1,F_0)\overset{n\to\infty}{\to} 0$ in Lemma \ref{lemma.tsybakov} (with properly chosen $\Theta$ and $\zeta$). Hence, by Markov's inequality,
\begin{align*}
\inf_{\hat{F}} \sup_{P\in\calM_S} \bE_P |\hat{F}- F(P)| \ge s\cdot \inf_{\hat{F}} \sup_{P\in\calM_S} \bP_P \left(|\hat{F}- F(P)| \ge s\right) \ge \frac{s}{2} \gtrsim \sqrt{\frac{S}{n\ln n}}
\end{align*}
which is a desired contradiction to \eqref{eq.functional_estimation}!

\subsection{Lower Bound on the ``Variance"}
To establish the second lower bound, we will essentially reduce the sorted distribution estimation problem to the traditional distribution estimation where labels are required. Specifically, we consider the scenario where we have known a priori that the probability vector is sorted, in which case an accurate estimator for $P^<$ can be easily transformed into an accurate estimator for $P$. Then in this scenario, we recover the traditional $\Omega(\sqrt{\frac{S}{n}})$ lower bound for estimating $P$. The reason why this lower bound does not hold for large $S$ is that when $S$ exceeds some threshold, the prior knowledge that the probability vector is sorted starts to make adjacent entries become informative on the inference of the entry in the middle.

For some constant $c>0$ large enough, define
\begin{align*}
S' \triangleq S \wedge \left(\frac{n}{c}\right)^{\frac{1}{3}}
\end{align*}
as the new support size, and without loss of generality we assume that $S'=2T+1$ for some integer $T$. Fixing some $\lambda>0$ to be determined later, we associate a probability vector $P_{\bm{\epsilon}}\in\calM_S$ to any binary vector $\bm{\epsilon}=(\epsilon_1,\epsilon_2,\cdots,\epsilon_T)\in\{\pm 1\}^T$ as follows:
\begin{align*}
P_{\bm{\epsilon}} &\triangleq ( 0, 0, \cdots, 0, x_1+\lambda\epsilon_1,\cdots,x_T+\lambda\epsilon_T,y_1-\lambda\epsilon_1,\cdots,y_T-\lambda\epsilon_T, 1-\sum_{i=1}^T (x_i+y_i))\\
&\text{with} \qquad x_i\triangleq \frac{c(t+i)^2}{100n}, \qquad y_i\triangleq \frac{c(t+T+i)^2}{100n}, \qquad \text{and} \qquad t\triangleq\left(\frac{n}{cT}\right)^{\frac{1}{2}}.
\end{align*}
Hence, as long as
\begin{align}\label{eq:condition_lambda}
\lambda \in \left(0, \frac{ct}{100n}\right),
\end{align}
it is easy to check that $P_{\bm{\epsilon}}$ is sorted in an ascending order for any $\bm{\epsilon}\in \{\pm 1\}^T$. As a result, restricting to the subclass $P\in \calP\triangleq\{P_{\bm{\epsilon}}: \bm{\epsilon}\in \{\pm 1\}^T\}\subset \calM_S$, estimating $P^<$ is equivalent to estimating $P$. In other words,
\begin{align}\label{eq.reduction}
\inf_{\hat{P}}\sup_{P\in \calM_S} \bE_P\|\hat{P}-P^<\|_1 \ge \inf_{\hat{P}}\sup_{P\in \calP} \bE_P\|\hat{P}-P\|_1.
\end{align}

Next we lower bound the RHS of \eqref{eq.reduction}. Consider a uniform prior on $\bm{\epsilon}\in \{\pm 1\}^T$, the fact that the Bayes risk under any prior is always a lower bound for the minimax risk gives that
\begin{align*}
\inf_{\hat{P}}\sup_{P\in \calP} \bE_P\|\hat{P}-P\|_1 \ge \inf_{\hat{P}}\bE_{\bm{\epsilon}} \bE_{P_{\bm{\epsilon}}}\|\hat{P}-P_{\bm{\epsilon}}\|_1 \ge \frac{T\lambda}{10}\cdot \inf_{\hat{P}} \bP\left(\|\hat{P}-P_{\bm{\epsilon}}\|_1\ge \frac{T\lambda}{10}\right).
\end{align*}
Defining the test function $\hat{\bm{\epsilon}}\triangleq\arg\min_{\bm{\epsilon}\in \{\pm 1\}^T} \|\hat{P}-P_{\bm{\epsilon} }\|_1$, by triangle inequality it is straightforward to see that the event $d_{\text{H}}(\hat{\bm{\epsilon}},\bm{\epsilon})\ge \frac{T}{5}$ implies $\|\hat{P}-P_{\bm{\epsilon}}\|_1\ge \frac{T\lambda}{10}$, where $d_{\text{H}}(\cdot,\cdot)$ is the Hamming metric $d_{\text{H}}(x,y)=\sum_{i=1}^T \mathbbm{1}(x_i\neq y_i)$. Consequently, we further have
\begin{align}\label{eq.eps_hat}
\inf_{\hat{P}}\sup_{P\in \calP} \bE_P\|\hat{P}-P\|_1 \ge \frac{T\lambda}{10}\cdot \inf_{\hat{\bm{\epsilon}}}\bP\left(d_{\text{H}}(\hat{\bm{\epsilon}},\bm{\epsilon})\ge \frac{T}{5}\right).
\end{align}
To lower bound the RHS of \eqref{eq.eps_hat}, we introduce the distance-based Fano's inequality as follows:
\begin{lemma}\cite[Corollary 1]{duchi2013distance}\label{lemma.fano}
	Let random variables $V$ and $\hat{V}$ take value in $\calV$, $V$ be uniform on some finite $\calV$, and $V-X-\hat{V}$ form a Markov chain. Let $d$ be any metric on $\calV$, and for $t>0$, define
	\begin{align*}
	N_{\max}(t) \triangleq \max_{v\in \calV} |v'\in V: d(v,v')\le t|, \qquad N_{\min}(t) \triangleq \min_{v\in \calV} |v'\in V: d(v,v')\le t|.
	\end{align*}
	If $N_{\max}(t)+N_{\min}(t)<|\calV|$, the following inequality holds:
	\begin{align*}
	\bP(d(V,\hat{V})>t) \ge 1 - \frac{I(V;X)+\ln 2}{\ln\frac{|\calV|}{N_{\max}(t)}}
	\end{align*}
	where $I(V;X)\triangleq \bE_{P_{V,X}}\left[\ln \frac{dP_{V,X}}{dP_V\times dP_X}\right]$ denotes the mutual information between $V$ and $X$.
\end{lemma}

Applying Lemma \ref{lemma.fano} to the Markov chain $\bm{\epsilon}-X-\hat{\bm{\epsilon}}$ with Hamming metric and $t=\frac{T}{5}$, by Lemma \ref{lemma.poissontail} we know that $\frac{N_{\max}(t)}{|\calV|}\le \exp(-\frac{T}{8})$. Moreover, in the Poissonized model we have
\begin{align*}
I(\bm{\epsilon};X) &\stepa \le \bE_{\bm{\epsilon}} D(P_{X|\bm{\epsilon}}\|P_0) \\
&= \sum_{i=1}^T \left(\bE_{\bm{\epsilon}} D(\spo(n(x_i+\lambda\epsilon_i)) \|\spo(nx_i) ) + \bE_{\bm{\epsilon}} D(\spo(n(y_i-\lambda\epsilon_i)) \|\spo(ny_i) )\right) \\
&\stepb \le n\lambda^2\cdot \sum_{i=1}^T \left(\frac{1}{x_i} + \frac{1}{y_i}\right) \\
&\stepc \le 200T\cdot n\lambda^2
\end{align*}
where $P_0$ is the probability measure $P_{\bm{\epsilon}}$ applied to $\bm{\epsilon}={\bf 0}$, (a) follows from the variational representation of mutual information
\begin{align*}
I(X;Y) = \inf_{Q_Y} \bE_{P_X} D(P_{Y|X}\|Q_Y), 
\end{align*}
(b) follows from $D(\spo(\lambda_1)\|\spo(\lambda_2))=\lambda_1\ln\frac{\lambda_1}{\lambda_2}-\lambda_1+\lambda_2\le \frac{(\lambda_1-\lambda_2)^2}{\lambda_2}$, and (c) follows from $x_i, y_i\ge (200T)^{-1}$. Consequently, a combination of \eqref{eq.reduction}, \eqref{eq.eps_hat} and Lemma \ref{lemma.fano} yields
\begin{align}\label{eq.fano}
\inf_{\hat{P}}\sup_{P\in \calM_S} \bE_P\|\hat{P}-P^<\|_1 \ge \frac{T\lambda}{10}\left(1-\frac{800nT^2\lambda^2+\ln 2}{T/8}\right).
\end{align}

Choosing $\lambda=\frac{c'}{\sqrt{nT}}$, for $c'$ small enough the condition \eqref{eq:condition_lambda} is fulfilled, and by \eqref{eq.fano} and the choice of $T$ we conclude that
\begin{align*}
\inf_{\hat{P}}\sup_{P\in \calM_S} \bE_P\|\hat{P}-P^<\|_1 \gtrsim \sqrt{\frac{T}{n}} \gtrsim \sqrt{\frac{S}{n}}\wedge n^{-\frac{1}{3}},
\end{align*}
establishing the second term of Theorem \ref{thm.lower}.

\section{Auxiliary Lemmas}

\begin{lemma}\cite[Lemma 28]{Han--Jiao--Weissman2016minimaxdivergence}\label{lem.polycoeff}
	Let $p_n(x) = \sum_{\nu=0}^n a_\nu x^\nu$ be a polynomial of degree at most $n$ such that $|p_n(x)|\leq A$ for $x\in [a,b]$. Then
	\begin{enumerate}
		\item If $a+b\neq 0$, then
		\begin{align*}
		|a_\nu| \le 2^{7n/2}A\left|\frac{a+b}{2}\right|^{-\nu}\left(\left|\frac{b+a}{b-a}\right|^n +1 \right), \qquad \nu=0,\cdots,n.
		\end{align*}
		\item If $a+b = 0$, then
		\begin{align*}
		|a_\nu| \leq A b^{-\nu} (\sqrt{2}+1)^n, \qquad \nu=0,\cdots,n.
		\end{align*}
	\end{enumerate}
\end{lemma}

\begin{lemma}\label{lem_hoeffding}
	\cite{Hoeffding1963probability} For independent and identically distributed random variables $X_1,\cdots,X_n$ with $a\le X_i\le b$ for $1\le i\le n$, denote $S_n=\sum_{i=1}^n X_i$, we have for any $t>0$,
	\begin{align*}
	\bP\left\{|S_n-\bE[S_n]|\ge t\right\} \le 2\exp\left(-\frac{2t^2}{n(b-a)^2}\right).
	\end{align*}
\end{lemma}

\begin{lemma}\label{lemma.poissontail}
	\cite[Theorem 5.4]{mitzenmacher2005probability}
	For $X\sim \mathsf{Poi}(\lambda)$ or $X\sim \mathsf{B}(n,\frac{\lambda}{n})$ and any $\delta>0$, we have
	\begin{align*}
	\mathbb{P}(X\ge (1+\delta)\lambda) &\le \left(\frac{e^\delta}{(1+\delta)^{1+\delta}}\right)^\lambda \le \exp(-\frac{(\delta^2\wedge \delta)\lambda}{3}),\\
	\mathbb{P}(X\le (1-\delta)\lambda) &\le \left(\frac{e^{-\delta}}{(1-\delta)^{1-\delta}}\right)^\lambda \le \exp(-\frac{\delta^2\lambda}{2}).
	\end{align*}
\end{lemma}


\section{Proof of Main Lemmas}
\subsection{Proof of Corollary \ref{cor.MLE}}
By Theorem \ref{thm.main}, it suffices to prove that for the empirical distribution $P_n$, we have
\begin{align*}
\sup_{P\in \calM_S} \bE_P \|P_n^<-P^<\|_1 \asymp \sqrt{\frac{S}{n}}, \qquad n\ge S.
\end{align*}

By \cite{han2015minimax}, the upper bound follows from 
\begin{align*}
\sup_{P\in \calM_S} \bE_P \|P_n^<-P^<\|_1 \le \sup_{P\in \calM_S} \bE_P \|P_n-P\|_1 \le \sqrt{\frac{S}{n}}.
\end{align*}

For the lower bound, note that \cite[Theorem 1]{berend2013sharp} shows that for $n\ge S$, 
\begin{align*}
\bE \left|\frac{1}{n}\mathsf{B}(n,\frac{1}{S}) - \frac{1}{S}\right| \ge \sqrt{\frac{S-1}{2nS^2}}.
\end{align*}
Consider $P=(\frac{1}{S},\cdots,\frac{1}{S})$ to be the uniform distribution, then there is no difference between estimating $P^<$ and estimating $P$. Hence,
\begin{align*}
\bE_P \|P_n^<-P^<\|_1 = \bE_P \|P_n-P\|_1 \ge S\cdot\bE \left|\frac{1}{n}\mathsf{B}(n,\frac{1}{S}) - \frac{1}{S}\right| \ge \sqrt{\frac{S-1}{2n}}
\end{align*}
as desired. 

\subsection{Proof of Lemma \ref{lemma.wasserstein}}
Without loss of generality assume that $P,Q$ are sorted in an ascending order. Consider the coupling on $(X,Y)$ which is uniformly distributed on the multiset $\{(p_1,q_1),\cdots,(p_S,q_S)\}$, we immediately have
\begin{align*}
S\cdot W(\mu_P,\mu_Q) \le \|P-Q\|_1.
\end{align*}

For the opposite inequality, let $r_{ij}\triangleq \bP(X=p_i, Y=q_j)$, it's straightforward to see that
\begin{align*}
W(\mu_P,\mu_Q) = \min & ~ \sum_{i=1}^S\sum_{j=1}^S r_{ij}|p_i-q_j| \\
\text{s.t.} & ~ \sum_{i=1}^S r_{ij}=\frac{1}{S}, \qquad j=1,2,\cdots,S \\
& ~ \sum_{j=1}^S r_{ij}=\frac{1}{S}, \qquad i=1,2,\cdots,S \\
& ~ r_{ij}\ge 0, \qquad i,j=1,2,\cdots,S.
\end{align*}

For linear programming, there must be a vertex of the simplex which attains the minimum of the objective. In other words,
\begin{align*}
S\cdot W(\mu_P,\mu_Q) =  \min_{\sigma} \sum_{i=1}^S |p_i-q_{\sigma(i)}|
\end{align*}
where $\sigma$ ranges over all permutations on $\{1,\cdots,S\}$. Finally, note that for $i<j, k<l$, we have (by symmetry we assume that $p_i\le q_k$)
\begin{align*}
|p_i-q_l| + |p_j-q_k| &= |p_i-q_k| + |q_k-q_l| + |p_j-q_k| \ge |p_i-q_k| + |p_j-q_l|.
\end{align*}
In other words, switching $\sigma(i)$ and $\sigma(j)$ whenever $i<j$ and $\sigma(i)>\sigma(j)$ can only make the value of the objective smaller. Hence, the minimum is attained at $\sigma=\text{id}$, and $S\cdot W(\mu_P,\mu_Q)\ge \|P-Q\|_1$, as desired. 

\subsection{Proof of Lemma \ref{lemma.randomization}}
We prove Lemma \ref{lemma.randomization} via figure. In the following figure, all curves represent different CDFs. Note that $\mu_P$ is a discrete distribution supported on $S=4$ elements, and $\mu$ is an arbitrary distribution. The area of the yellow region exactly represents the Wasserstein distance $W(\mu_P,\mu)$. The randomization procedure picks up one point uniformly at random from each small interval $[0,\frac{1}{4}], [\frac{1}{4},\frac{1}{2}], [\frac{1}{2},\frac{3}{4}],[\frac{3}{4},1]$ on the $y$-axis, and then returns the corresponding inverse on the $x$-axis. Now from a vertical viewpoint, it is straightforward to verify that $\bE W(\mu_P,\mu_Q)$ is also the yellow area, as desired. 
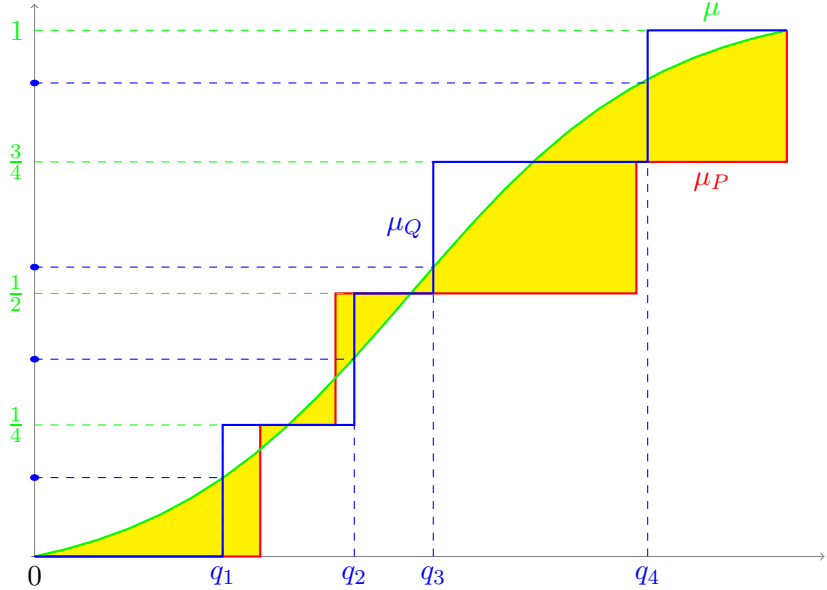
\begin{center}\label{fig.randomization}
	\centering
	\begin{tikzpicture}[xscale=5, yscale=3.5]
	\draw [thick, <->, help lines] (0,2.1) -- (0,0) -- (2.1,0);
	\node [below] at (0,0) {$0$};
	\draw [fill=yellow, domain=0:2] plot(\x, { 2/(1-exp(3)) +2*(1-2/(1-exp(3)))/( 1+exp( 3*(-\x+1) ) ) }) -- (2,1.5) -- (1.6,1.5) -- (1.6,1) -- (0.8,1) -- (0.8,0.5) -- (0.6,0.5) -- (0.6,0) to cycle;
	\draw [thick, red] (0,0) -- (0.6,0) -- (0.6,0.5) -- (0.8,0.5) -- (0.8,1) -- (1.6,1) -- (1.6,1.5) -- (2,1.5) -- (2,2); 
	\draw [thick, green, domain=0:2] plot(\x, { 2/(1-exp(3)) +2*(1-2/(1-exp(3)))/( 1+exp( 3*(-\x+1) ) ) });
	\node [green, above] at (1.8,2) {$\mu$};
	\node [red, below] at (1.8,1.5) {$\mu_P$};
	
	\node [green, left] at (0,0.5) {$\frac{1}{4}$};	
	\node [green, left] at (0,1) {$\frac{1}{2}$};	
	\node [green, left] at (0,1.5) {$\frac{3}{4}$};	
	\node [green, left] at (0,2) {$1$};	
	\draw [dashed, green] (0,0.5) -- (0.6,0.5);
	\draw [dashed, green] (0,1) -- (1,1);
	\draw [dashed, green] (0,1.5) -- (1.32,1.5);
	\draw [dashed, green] (0,2) -- (2,2);
	
	\draw [blue, fill=blue] (0,0.3) circle (0.3pt);
	\draw [blue, fill=blue] (0,0.75) circle (0.3pt);
	\draw [blue, fill=blue] (0,1.1) circle (0.3pt);
	\draw [blue, fill=blue] (0,1.8) circle (0.3pt);
	
	\draw [blue, dashed] (0,0.3) -- (0.5,0.3);
	\draw [blue, dashed] (0,0.75) -- (0.85,0.75);
	\draw [blue, dashed] (0,1.1) -- (1.06,1.1);
	\draw [blue, dashed] (0,1.8) -- (1.63,1.8);
	
	\draw [blue, thick] (0,0) -- (0.5,0) -- (0.5,0.5) -- (0.85,0.5) -- (0.85,1) -- (1.06,1) -- (1.06,1.5) -- (1.63,1.5) -- (1.63,2) -- (2,2);
	\node [blue, left] at (1.06,1.25) {$\mu_Q$};
	\node [blue, below] at (0.5,0) {$q_1$};
	\node [blue, below] at (0.85,0) {$q_2$};
	\node [blue, below] at (1.06,0) {$q_3$};
	\node [blue, below] at (1.63,0) {$q_4$};
	\draw [blue, dashed] (0.5,0) -- (0.5,0.5);
	\draw [blue, dashed] (0.85,0) -- (0.85,0.5);
	\draw [blue, dashed] (1.06,0) -- (1.06,1);
	\draw [blue, dashed] (1.63,0) -- (1.63,1.5);
	\end{tikzpicture}
\end{center}

\subsection{Proof of Lemma \ref{lemma.poissonization}}
Let $R(n,S,\pi), R_P(n,S,\pi)$ be the corresponding Bayes risks under prior $\pi$ in the Multinomial and Poissonized models, respectively. By \cite{Jiao--Venkat--Han--Weissman2015minimax}, 
\begin{align*}
R_P(n,S,\pi) = \sum_{m=0}^\infty R(m,S,\pi)\cdot \bP(\spo(n)=m).
\end{align*}
Note that $R(m,S,\pi)$ is non-increasing in $m$, by Lemma \ref{lemma.poissontail} and Markov's inequality we have
\begin{align*}
R_P(n,S,\pi) &\le 2\cdot \bP(\spo(n)<\frac{n}{2}) + R(\frac{n}{2},S,\pi) \le 2\exp(-\frac{n}{8}) + R(\frac{n}{2},S,\pi), \\
R_P(n,S,\pi) &\ge R(2n,S,\pi)\cdot \bP(\spo(n)\le 2n) \ge R(2n,S,\pi)\cdot\left(1-\frac{\bE\spo(n)}{2n}\right) = \frac{1}{2}R(2n,S,\pi).
\end{align*}
Now taking supremum over $\pi$ and using the minimax theorem \cite{Wald1950statistical} complete the proof.

\subsection{Proof of Lemma \ref{lemma.feasibility}}
By definition of $\mu_{P,j}$ and $A_j$, the claimed result is equivalent to
\begin{align*}
\mathbb{P}\left(\left|\sum_{i\in A_j} \left(g_{k,x_j}(\hat{p}_{i,2}) - (p_i-x_j)^k \right)\right| > \sqrt{S_j\ln n}\cdot \left(\frac{c_3j\ln n}{n}\right)^{k} \right) \le 2n^{-4}.
\end{align*}

Let $Z_i\triangleq g_{k,x_j}(\hat{p}_{i,2}) - (p_i-x_j)^k$, and $M_i\triangleq \sup_{\hat{p}_{i,2}: |\hat{p}_{i,2}-x_j|\le 3c_1j\ln n/n} |Z_i|$. We first establish an upper bound for $M_i$ with $i\in A_j$: firstly, the condition $i\in A_j$ implies that $p_i\in \tilde{I}_j$, and thus
\begin{align}\label{eq.M_upper_1}
|p_i-x_j|^k \le \left|\frac{c_1(j+1)^2\ln n}{n}-\frac{c_1j(j-1)\ln n}{n}\right|^k \le \left(\frac{4c_1j\ln n}{n}\right)^k.
\end{align}
To upper bound $|g_{k,x_j}(\hat{p}_{i,2})|$, we introduce the following lemma:
\begin{lemma}\label{lemma.charlier}
	Let $k\ge 1, np\in \mathbb{N}$ and $x\in [0,1]$. If $\max\{|x-p|, \sqrt{\frac{4pk}{n}}\}\le \Delta$, we have
	\begin{align*}
	|g_{k,x}(p)| = \left|\sum_{l=0}^k \binom{k}{l}(-x)^{k-l}\prod_{l'=0}^{l-1}\left(p-\frac{l'}{n}\right)\right| \le (2\Delta)^k.
	\end{align*}
\end{lemma}
The proof of Lemma \ref{lemma.charlier} is postponed to the end of this subsection. The conditions of Lemma \ref{lemma.charlier} are fulfilled by $\Delta=\frac{3c_1j\ln n}{n}$ as long as $c_1>\frac{4}{3}c_2$, and thus
\begin{align}\label{eq.M_upper_2}
|g_{k,x_j}(\hat{p}_{i,2})| \le \left(\frac{6c_1j\ln n}{n}\right)^k.
\end{align}
As a result, a combination of \eqref{eq.M_upper_1} and \eqref{eq.M_upper_2} ensures that
\begin{align}\label{eq.M_upper}
M_i \le \left(\frac{10c_1j\ln n}{n}\right)^k, \qquad \forall i\in A_j.
\end{align}

Define a new random variable $\tilde{Z}_i\triangleq \max\{\min\{Z_i,M_i\},-M_i\}$ as the truncated version of $Z_i$, the Hoeffding's inequality is about to be applied to the independent and bounded $\tilde{Z}_i$. We need to show that $\tilde{Z}_i$ and $Z_i$ are indeed close in expectation. Clearly,
\begin{align*}
|\bE (\tilde{Z}_i-Z_i)| &\le \bE|g_{k,x_j}(\hat{p}_{i,2})|\mathbbm{1}(|\hat{p}_{i,2}-x_j|>\Delta_j) \\
&= \sum_{m: |m-nx_j|>n\Delta_j} |g_{k,x_j}(\frac{m}{n})|\cdot\bP(\spo(np_i)=m) \\
&\le \sum_{m: |m-nx_j|>n\Delta_j} 2^k\left|\frac{m}{n}-x_j\right|^k\cdot\bP(\spo(np_i)=m)
\end{align*}
where $\Delta_j\triangleq\frac{3c_1j\ln n}{n}$ and we have used Lemma \ref{lemma.charlier} in the last step. If $m+1>n(x_j+\Delta_j)$, we have
\begin{align*}
\frac{\bP(\spo(np_i)=m+1)}{\bP(\spo(np_i)=m)} = \frac{np_i}{m+1} \le \frac{c_1(j+1/2)^2\ln n}{c_1j(j+2)\ln n} \le 1 - \frac{1}{4j}.
\end{align*}
Let $m_{\max}$ be the largest integer such that $m_{\max}\le n(x_j+\Delta_j)$, by Lemma \ref{lemma.poissontail} and choosing $c_1>0$ large enough (as in Lemma \ref{lemma.concentration}) we have $\bP(\spo(np_i)=m_{\max})\le n^{-5}$. Hence,
\begin{align*}
\sum_{m: m-nx_j>n\Delta_j} \left|\frac{m}{n}-x_j\right|^k\cdot\bP(\spo(np_i)=m) &\le \sum_{l=0}^\infty (\Delta_j+\frac{l}{n})^k\cdot n^{-5}\left(1-\frac{1}{4j}\right)^l \\
&\le n^{-5}\Delta_j^k\cdot \sum_{l=0}^\infty \exp\left(\frac{kl}{n\Delta_j} - \frac{l}{4j}\right) \\
&\le n^{-5}\Delta_j^k\left[1-\exp\left(-\frac{1}{4j}+\frac{k}{n\Delta_j}\right)\right]^{-1}\\
&\le n^{-5}\Delta_j^k\cdot \left[1-\exp\left(-\frac{1}{M}\left(\frac{1}{4}-\frac{c_2}{3c_1}\right)\right)\right]^{-1} \le cMn^{-5}\Delta_j^k
\end{align*}
where $c$ is a constant depending only on $c_1,c_2$ when $3c_1>4c_2$. 

The case where $m<n(x_j-\Delta_j)$ can be handled using similar arguments. As a result, as long as $c_1>2c_2$, we have
\begin{align}\label{eq.mean_diff}
|\bE (\tilde{Z}_i-Z_i)| \le \frac{cM}{n^5}\cdot \left(\frac{6c_1j\ln n}{n}\right)^k.
\end{align}
Note that $\bE Z_i=0$ by \eqref{eq.unbiased}, and $|\tilde{Z}_i|\le M_i\le \left(\frac{9c_1j\ln n}{n}\right)^k$ by \eqref{eq.M_upper}, Hoeffding's inequality (cf. Lemma \ref{lem_hoeffding}) with \eqref{eq.mean_diff} yields
\begin{align*}
&\bP\left(\left|\sum_{i\in A_j} Z_i\right| > \sqrt{S_j\ln n}\cdot \left(\frac{c_3j\ln n}{n}\right)^{k}\right) \nonumber \\
&\le \bP\left(\left|\sum_{i\in A_j} (\tilde{Z}_i-\bE \tilde{Z}_i)\right| > \sqrt{S_j\ln n}\cdot \left(\frac{c_3j\ln n}{n}\right)^{k} - \frac{24MS_j}{n^5}\cdot \left(\frac{6c_1j\ln n}{n}\right)^k\right) \\
&\le 2\exp\left(-\frac{\left(\sqrt{S_j\ln n}\cdot \left(\frac{c_3j\ln n}{n}\right)^{k} - \frac{cMS_j}{n^5}\cdot \left(\frac{6c_1j\ln n}{n}\right)^k\right)^2}{2S_j\left(\frac{10c_1j\ln n}{n}\right)^{2k}}\right) \le 2n^{-4}
\end{align*}
as long as $c_3>30c_1$, as desired. 

\begin{proof}[Proof of Lemma \ref{lemma.charlier}]
	The Charlier polynomial $c_k(u,a)$ for $u\in \mathbb{N}$ is defined as
	\begin{align*}
	c_k(u,a) \triangleq \sum_{l=0}^k (-1)^{k-l}\binom{k}{l}\frac{(u)_l}{a^l}
	\end{align*}
	where $(u)_l\triangleq u(u-1)\cdots(u-l+1)$ is the falling factorial. For the Charlier polynomial, the following identities hold \cite{peccati2011some}, \cite[Eqn. (574)]{jiao2017minimax}:
	\begin{align*}
	\sum_{k=0}^\infty \frac{c_k(u,a)}{k!}t^k &= e^{-t}\left(1+\frac{t}{a}\right)^u, \\
	\sum_{l=0}^k \binom{k}{l}(-b)^{k-l}a^lc_l(u,a)&=(a+b)^kc_k(u,a+b).
	\end{align*}
	
	The function $g_{k,x}(p)$ is related to the Charlier polynomial via the identity $g_{k,x}(p)=x^kc_k(np,nx)$, and thus the previous identities translate into the following:
	\begin{align}
	\sum_{k=0}^\infty \frac{g_{k,p}(p)}{k!}t^k &= e^{-pt}(1+\frac{t}{n})^{np}, \label{eq.mgf} \\
	g_{k,x}(p) &= \sum_{l=0}^k \binom{k}{l}(p-x)^{k-l}g_{l,p}(p). \label{eq.sum_rule}
	\end{align}
	We can rewrite \eqref{eq.mgf} into the following form:
	\begin{align*}
	\sum_{k=0}^\infty \frac{g_{k,p}(p)}{k!}t^k &= \left[e^{-\frac{t}{n}}\left(1+\frac{t}{n}\right)\right]^{np} = \left[\left(\sum_{l=0}^\infty \frac{1}{l!}(-\frac{t}{n})^l\right)\left(1+\frac{t}{n}\right)\right]^{np}\\
	&=\left(1-\sum_{l=2}^\infty \frac{l-1}{l!}(-\frac{t}{n})^l\right)^{np}.
	\end{align*}
	Comparing the coefficients of $t^k$ at both sides yields
	\begin{align*}
	|g_{k,p}(p)| &\le \frac{k!}{n^k}\cdot \sum_{1\le r\le k/2}\binom{np}{r}\sum_{\sum_{i=1}^r l_i=k, l_i\ge 2}  \prod_{i=1}^r\frac{l_i-1}{l_i!} \\
	&\le \frac{k!}{n^k}\cdot \sum_{1\le r\le k/2}\binom{np}{r}\sum_{\sum_{i=1}^r l_i=k, l_i\ge 2}  1 \\
	&= \frac{k!}{n^k}\cdot \sum_{1\le r\le k/2}\binom{np}{r} \binom{k-r-1}{r-1} \\
	&\le \frac{k!}{n^k}\cdot \sum_{1\le r\le k/2}\binom{np}{r}2^{k-r-1}.
	\end{align*}
	
	We distinguish into two cases: if $np\ge k$, we have $\binom{np}{r}\le \binom{np}{k/2}$, and thus
	\begin{align*}
	|g_{k,p}(p)| \le \frac{k!}{n^k}\cdot\binom{np}{k/2}2^k \le \left(\frac{4pk}{n}\right)^{\frac{k}{2}}.
	\end{align*}
	If $np<k$, we use the inequality $\binom{np}{r}\le 2^{np}\le 2^k$ to upper bound $|g_{k,p}(p)|$ as
	\begin{align*}
	|g_{k,p}(p)| \le \frac{k!}{n^k}\cdot 2^{2k} = \left(\frac{4k}{n}\right)^k.
	\end{align*}
	
	Combining these two cases, and using the assumption on $\Delta$ and \eqref{eq.sum_rule}, we conclude that
	\begin{align*}
	|g_{k,x}(p)| \le \sum_{l=0}^k \binom{k}{l}\Delta^{k-l}\cdot \Delta^l = (2\Delta)^k
	\end{align*}
	as desired.
\end{proof}

\subsection{Proof of Lemma \ref{lemma.functional}}
By Corollary \ref{cor.bad_event}, the conditions of Lemma \ref{lemma.lmm} are satisfied with probability at least $1-Sn^{-5}-2Mn^{-4}\ge 1-3Sn^{-4}$. Let $\mu_{P,j}$ and $\mu_j$ be given in \eqref{eq.mu_P_j} and \eqref{eq.LP_1}, \eqref{eq.LP_2} respectively, we have
\begin{align*}
|\hat{F} - F(P)| &\le \sum_{j=1}^M \left|\int_{\reals} f(x)(\mu_j(dx)-\mu_{P,j}(dx))\right|.
\end{align*}

By \eqref{eq.decompose}, for any degree-$K$ polynomial $P_j(x)=\sum_{k=0}^K a_{k,j}x^k$ on $\tilde{I}_j$, we further have
\begin{align*}
|\hat{F} - F(P)| \le \sum_{j=1}^M \left(\int_\reals |f(x)-P_j(x)|(\mu_j(dx)+\mu_{P,j}(dx)) + \sum_{k=1}^K |a_{k,j}|\cdot 2\sqrt{S_j\ln n}\left(\frac{c_3j\ln n}{n}\right)^k\right).
\end{align*}
By assumption, $|P_j(x)-f(x_j)|\le \|f-P\|_{\infty, \tilde{I}_j} + |f(x)-f(x_j)|\le 2M_j$ for any $j$, by Lemma \ref{lem.polycoeff} we have
\begin{align*}
|a_{k,j}| \le 2^{6K+3}M_j\left(\frac{c_1j\ln n}{5n}\right)^{-k}, \qquad k\ge 1.
\end{align*}
As a result, the second term can be upper bounded as
\begin{align*}
\sum_{k=1}^K |a_{k,j}|\cdot 2\sqrt{S_j\ln n}\left(\frac{c_3\ln n}{n}\right)^k &\le \sum_{k=1}^K 2^{6K+3}M_j\left(\frac{c_1j\ln n}{5n}\right)^{-k}\cdot 2\sqrt{S_j\ln n}\left(\frac{c_3j\ln n}{n}\right)^k \\
&= 16c_2n^{c_2(6\ln 2+\ln (5c_3/c_1))}(\ln n)^{\frac{3}{2}}\cdot M_j\sqrt{S_j} \lesssim n^{\epsilon}\cdot M_j\sqrt{S_j}
\end{align*}
as long as $c_2(6\ln 2+\ln(5c_3/c_1))<\epsilon$.  

For the corollary, note that Lemma \ref{lemma.feasibility} ensures that $\mu_1(\reals)\le S_1$ with high probability. Moreover, the condition $f(0)=0$ ensures that 
\begin{align*}
\inf_{P_1\in\mathsf{Poly}_K} \|f-P_1\|_{\infty,\tilde{I}_1} \le 2\cdot \inf_{P_1\in\mathsf{Poly}_K, P_1(0)=0} \|f-P_1\|_{\infty,\tilde{I}_1}.
\end{align*}
The proof is complete.

\subsection{Proof of Theorem \ref{thm.functional}}
We first consider the entropy functional $H(P)$ with $f(x)=-x\ln x$. By \cite{Jiao--Venkat--Han--Weissman2015minimax}, 
\begin{align*}
\inf_{P_1\in\mathsf{Poly}_K}\|f-P_1\|_{\infty, \tilde{I}_1} \asymp \frac{1}{n\ln n}.
\end{align*}
For $j\ge 2$, \cite[Theorem 7.2.1]{Ditzian--Totik1987} shows that
\begin{align*}
\inf_{P_j\in \mathsf{Poly}_K}\|f-P_j\|_{\infty,\tilde{I}_j} \le \frac{|\tilde{I}_j|^2}{2K^2}\sup_{x\in\tilde{I}_j} |f''(x)| \lesssim \frac{j^2}{n}\sup_{x\in\tilde{I}_j} |f''(x)|.
\end{align*}
For $f(x)=-x\ln x$, the previous inequality gives
\begin{align*}
\inf_{P_j\in \mathsf{Poly}_K}\|f-P_j\|_{\infty,\tilde{I}_j} \lesssim \frac{1}{n\ln n}, \qquad \forall j\ge 2.
\end{align*}

As a result, 
\begin{align}\label{eq.entropy_bias}
\sum_{j=1}^M S_j\cdot \inf_{P_j\in \mathsf{Poly}_K}\|f-P_j\|_{\infty,\tilde{I}_j} \lesssim \frac{1}{n\ln n}\sum_{j=1}^M S_j = \frac{S}{n\ln n}.
\end{align}
For the second term in Lemma \ref{lemma.functional}, we have $M_j\lesssim \frac{j(\ln n)^2}{n}$. Hence, by \eqref{eq.var_sum_bound}, 
\begin{align}\label{eq.entropy_var}
\sum_{j=1}^M n^{\epsilon}M_j\sqrt{S_j} \lesssim \sum_{j=1}^M \frac{j\sqrt{S_j}}{n^{1-\epsilon}} \lesssim n^{\epsilon}\left(\sqrt{\frac{S}{n}}\wedge n^{-\frac{1}{3}}\right).
\end{align}
The desired result for $H(P)$ now follows from \eqref{eq.entropy_bias}, \eqref{eq.entropy_var} and Lemma \ref{lemma.functional}. The results for $F_\alpha(P), 0<\alpha<1$ can also be obtained in a similar way.

Next we look at the support size functional $S(P)$ with $f(x)=\mathbbm{1}(x\neq 0)$. Here to apply Lemma \ref{lemma.functional}, it suffices to consider $j=1$ with the corresponding interval $\tilde{I}_1'=\{0\}\cup[\frac{1}{k},\frac{c_1\ln n}{n}]$. By \cite{wu2015chebyshev}, 
\begin{align*}
\inf_{P_1\in\mathsf{Poly}_K}\|f-P_1\|_{\infty, \tilde{I}_1'} \asymp \exp\left(-\Theta\left(\sqrt{\frac{n\ln k}{k}}\right)\right).
\end{align*}
In addition, $M_1=1$ and $M_j=0$ for any $j\ge 2$, by Lemma \ref{lemma.functional} we know that
\begin{align*}
\sup_{P\in\calD_k} \bE_P |\hat{S}-S(P)| &\lesssim k\cdot n^{-3} + S_1\cdot \exp\left(-\Theta\left(\sqrt{\frac{n\ln k}{k}}\right)\right) + n^{\epsilon} \sqrt{S_1} \\
&\lesssim k\left[\exp\left(-\Theta\left(\sqrt{\frac{n\ln k}{k}}\right)\right) + \frac{n^{\epsilon}}{\sqrt{k}}\right].
\end{align*}
The proof is complete.

\end{document}